\newif\iffrontiers

\frontiersfalse
\iffrontiers
\documentclass[utf8]{frontiersFPHY}
\setcitestyle{square}
\else
\documentclass[10pt, a4paper]{scrartcl}
\fi

\usepackage[utf8]{inputenc} 
\usepackage[english]{babel}

\usepackage{amssymb}
\usepackage{amsmath,mathtools}
\usepackage{amsthm}
\usepackage{dsfont}
\usepackage{mathrsfs}% for descent cone D
\usepackage{lmodern}
\usepackage{paralist}

\iffrontiers\else
\usepackage{changepage} % to change margins for part of the document
\usepackage[calcwidth = .9\linewidth,bf, format=plain, font={small,sf}]{caption}
\fi

\usepackage{graphicx}
\graphicspath{{pics/}{./}}

\iffrontiers
\else
\usepackage{authblk}%for author affiliations

\usepackage[table,usenames,dvipsnames]{xcolor}
\usepackage[numbers,sort&compress]{natbib} 
\fi

\usepackage{algorithm}
\usepackage{algpseudocode} % algoirhtm State etc 

\usepackage[colorlinks=true]{hyperref}
\PassOptionsToPackage{linktocpage}{hyperref} %link numbers in TOC for arxiv
\hypersetup{
  colorlinks  = true, %Colour links instead of ugly boxes
  urlcolor    = violet, %Colour for external hyperlinks
  linkcolor   = NavyBlue, %Colour of internal links
  citecolor   = Green, %Colour of citations
  hypertexnames = false,
}

\newtheorem{theorem}{Theorem}
\newtheorem{corollary}[theorem]{Corollary}
\newtheorem{lemma}[theorem]{Lemma}
\newtheorem{proposition}[theorem]{Proposition}

\theoremstyle{definition}
\DeclareMathOperator{\DC}{\mathscr{D}}
\newcommand{\e}{\ensuremath\mathrm{e}}

\DeclareMathOperator{\Tr}{Tr}

\DeclareMathOperator{\rank}{rank}

\DeclareMathOperator{\supp}{supp}

\DeclareMathOperator{\cl}{cl}

\DeclareMathOperator{\cone}{cone}
\DeclareMathOperator{\conv}{conv}

\DeclareMathOperator{\argmin}{arg\min}

\newcommand{\reg}{\mathrm{reg}}

\DeclareMathOperator{\circone}{circ} %circular cone

\DeclareMathOperator{\vect}{vec}
\newcommand{\eff}{\mathrm{eff}}

\newcommand{\CC}{\mathbb{C}}
\newcommand{\RR}{\mathbb{R}}
\newcommand{\KK}{\mathbb{K}}% for \RR or \CC %overwritten with \RR for this paper

\newcommand{\1}{\mathds{1}}
\newcommand{\EE}{\mathbb{E}}
\newcommand{\PP}{\mathbb{P}}

\newcommand{\M}{\mathbb{M}} 

\newcommand{\mc}[1]{\mathcal{#1}}
\newcommand{\A}{\mc{A}}

\newcommand{\R}{\mc{R}}

\newcommand{\argdot}{{\,\cdot\,}}

\newcommand{\myleft}{\mathopen{}\mathclose\bgroup\left}
\newcommand{\myright}{\aftergroup\egroup\right}
\newcommand{\ad}{\dagger}

\newcommand{\norm}[1]{\left\Vert #1 \right\Vert} %norm with variable height
\newcommand{\iiiNorm}[1]{{\left\vert\kern-0.25ex\left\vert\kern-0.25ex\left\vert #1 
    \right\vert\kern-0.25ex\right\vert\kern-0.25ex\right\vert}}
\newcommand{\OneNorm}[1]{\norm{#1}_{1}} %one norm

\newcommand{\TwoNorm}[1]{\norm{#1}_{2}} %two norm

\newcommand{\InfNorm}[1]{\norm{#1}_{\infty}} %Inf norm

\newcommand{\lOneNorm}[1]{\norm{#1}_{\ell_1}} %1 norm

\newcommand{\lTwoNorm}[1]{\norm{#1}_{\ell_2}} %2 norm

\newcommand{\lInfNorm}[1]{\norm{#1}_{\ell_\infty}} %inf norm

\newcommand{\fNorm}[1]{\lTwoNorm{#1}} % just as \ell_2-norm

\newcommand{\regNormMax}[2][\mu]{\norm{#2}_{#1,\mathrm{max}}} % max of norms

\newcommand{\regNormSum}[2][\lambda]{\norm{#2}_{#1, \mathrm{sum}}} % max of norms

\newcommand{\braket}[2]{\left\langle #1, #2 \right\rangle}

\newcommand{\totdim}{d}

\newcommand{\psuccmax}{p^\mathrm{max}_{\mathrm{success}}}

\newcommand{\deltamax}{\delta_{\max{}}}
\newcommand{\deltasum}{\delta_{\mathrm{sum}}}
\usepackage{xifthen}

\newcommand{\mytitle}{
  Simultaneous structures in convex signal recovery 
  \texorpdfstring{--}{-} 
  revisiting the convex combination of norms
}
\makeatletter % metadata for hyperref
 \hypersetup{pdftitle = {\mytitle},
       pdfauthor = {Martin Kliesch, Stanislaw J. Szarek, Peter Jung},
       pdfsubject = {Compressed sensing},
       pdfkeywords = {Convex analysis, 
                      sparse, low-rank, 
                      tensor reconstruction, 
                      compressive sensing,
                      convex algorithm,
                      statistical dimension, Gaussian distance,
                      RIP, restricted isometry property, 
                      circular cone,
                      polar, dual, maximum, sum, norm, 
                      semidefinite program, SDP,
                      negativity, 
                      }
      }
\makeatother

\newcommand{\ug}{%
  Institute of Theoretical Physics and Astrophysics, 
  University of Gda\'nsk,
  Poland%
}

\newcommand{\hhu}{%
  Institute for Theoretical Physics,
  Heinrich Heine University D{\"u}sseldorf, 
  Germany%
}

\newcommand{\paris}{%
  Institut de Math\'ematiques de Jussieu-PRG, 
  Sorbonne Universit\'e, 
  Paris, 
  France%
}

\newcommand{\cwru}{
  Department of Mathematics, Applied Mathematics and Statistics, 
  Case Western Reserve University, 
   Cleveland, 
   Ohio% 44106-7058, 
   ,
   USA%
}

\newcommand{\tub}{%
  Communications and Information Theory Group, 
  Technical University of Berlin, 
  Germany%
  }

\iffrontiers
\def\keyFont{\fontsize{8}{11}\helveticabold }
\def\firstAuthorLast{Kliesch {et~al.}} %use et al only if is more than 1 author
\def\Authors{Martin Kliesch\,$^{1,2,*}$, Stanislaw J. Szarek\,$^{3,4}$ and Peter Jung\,$^{5,*}$}
\else %iffrontiers
\fi %iffrontiers

\begin{document}
\iffrontiers
\onecolumn

\title[Simultaneous structures]{\mytitle}

\author[\firstAuthorLast ]{\Authors} %This field will be automatically populated

\address{} %This field will be automatically populated

\correspondance{} %This field will be automatically populated

% \extraAuth{}% If there are more than 1 corresponding author, comment this line and uncomment the next one.
\extraAuth{Peter Jung\\ peter.jung@tu-berlin.de}
\else
\title{\LARGE \mytitle}
\author[1,2]{Martin Kliesch}
\author[3,4]{Stanislaw J. Szarek}
\author[5]{Peter Jung}
\affil[1]{\hhu}
\affil[2]{\ug}
\affil[3]{\paris}
\affil[4]{\cwru}
\affil[5]{\tub}

\date{}
\fi
%%%=============================================

\maketitle

\begin{abstract}
  In compressed sensing one uses known structures of otherwise unknown signals to recover them from as few linear observations as possible.  The
  structure comes in form of some compressibility including different
  notions of sparsity and low rankness.  In many cases convex
  relaxations allow to efficiently solve the inverse problems using
  standard convex solvers at almost-optimal sampling rates.

  A standard practice to account for multiple simultaneous structures
  in convex optimization is to add further regularizers or
  constraints.  From the compressed sensing perspective there is then
  the hope to also improve the sampling rate. Unfortunately, when
  taking simple combinations of regularizers, this seems not
  to be automatically the case as it has been shown for several
  examples in recent works.

  Here, we give an overview over ideas of combining multiple structures in convex programs by taking weighted sums and weighted maximums. 
  We discuss explicitly cases where optimal weights are used reflecting an optimal tuning of the reconstruction. 
  In particular, we extend known lower bounds on the number of required measurements to the optimally weighted maximum by using geometric arguments. 
  As examples, we
  discuss simultaneously low rank and sparse matrices and notions of
  matrix norms (in the ``square deal'' sense) for regularizing for
  tensor products. 
  We state an SDP formulation for numerically estimating the statistical dimensions and find a tensor case where the lower bound is roughly met up to a factor of two. 

\iffrontiers
\tiny
 \keyFont{ \section{Keywords:} Compressed sensing, low rank, sparse, matrix, tensor, reconstruction, statistical dimension, convex relaxation} %All article types: you may provide up to 8 keywords; at least 5 are mandatory.
\end{abstract}
\else%iffrontiers
\end{abstract}
% \tableofcontents
\thispagestyle{empty}
\fi%iffrontiers

%%% ===================================================
\section{Introduction}
%%% ===================================================
  The recovery of an unknown signal from a limited number of 
  % indirect and sufficiently diverse 
  observations can be more efficient by exploiting compressibility and a priori known structure of the signal. 
  This compressed sensing
  methodology has been applied to many 
  fields in physics,
  applied math and engineering. 
  In the most common setting, the structure is given as sparsity in a known basis. 
  To mention some more
  recent directions, block-, group- and hierarchical sparsity,
  low-rankness in matrix or tensor recovery problems and also the
  generic concepts of atomic decomposition are important structures 
  present in many estimation problems.

  In most of these cases, convex relaxations render the inverse problems
  itself amenable to standard solvers at almost-optimal sampling
  rates and also show tractability from a theoretical perspective
  \cite{FouRau13}.  In one variant, one minimizes a regularizing
  function under the constraints given by the measurements.  A good
  regularizing function, or just regularizer, gives a strong
  preference in the optimization toward the targeted structure.
  For instance, the $\ell_1$-norm can be used to regularize for sparsity and the nuclear norm for low rankness of matrices.

  \subsection{Problem statement}
  Let us describe the compressed sensing setting in more detail.  We
  consider a linear \emph{measurement map} $\A: V \to \RR^m$ on a 
  $d$-dimensional vector space $V \cong \RR^{\totdim}$ define by its
  output components
\begin{align}
\A(x)_i \coloneqq \braket{a_i}{x} \, .
\end{align} 
Throughout this work we assume the measurements to be Gaussian, i.e., 
$\langle a_i, e_j \rangle \sim \mc N(0,1)$ iid., where $\{e_j\}_{j\in [d]}$ is an orthonormal basis of $V$.   
Moreover, we consider a given \emph{signal} $x_0 \in V$, which yields the (noiseless) \emph{measurement vector} 
\begin{equation}
  y \coloneqq \A(x_0) \, . 
\end{equation}
We wish to reconstruct $x_0$ from $\A$ and $y$ in a way that is computationally tractable. 
Following a standard compressed sensing approach, we consider a norm $\norm{\argdot}_\reg$ as regularizer. 
Specifically, we consider an outcome of the following convex
optimization problem
\begin{equation}\label{eq:rec_def}
\begin{aligned}
  x_\reg &\coloneqq \argmin \left\{ \norm{x}_\reg: \ \A(x) = y  \right\} 
\end{aligned}
\end{equation}
as a candidate for a reconstruction of $x_0$, where we will choose
$\norm{\argdot}_\reg$ to be either the weighted sum or the weighted
maximum of other norms \eqref{eq:def_max_sum_reg_norm}.  If
computations related to $\norm{\argdot}_\reg$ are also computationally
tractable, \eqref{eq:rec_def} can be solved efficiently.
We wish that $x_\reg$ coincides with $x_0$.  
Indeed, when
the number of measurements $m$ is large enough (compared to the
dimension $\totdim$) then, with high probability over the realization
of $\A$, the signal is reconstructed, i.e., $x_0 = x_\reg$.  For
instance, when $m \geq \totdim$ then $\A$ is invertible with
probability $1$ and, hence, the constraint $\A(x) = y$ in
\eqref{eq:rec_def} alone guarantees such a successful reconstruction.

However, when the signal is \emph{compressible}, 
i.e.,\ it can be described by fewer parameters than the dimension $\totdim$, 
then one can hope for a reconstruction from fewer measurements by choosing a good \emph{regularizer} $\norm{\argdot}_\reg$. 
For the case of Gaussian measurements, a quantity called the \emph{statistical dimension} 
gives the number of measurements sufficient \cite{Chandrasekaran2010}
and necessary \cite{AmeLotMcC14} (see also
\cite[Corollary~4]{MuHuaWri13}) for a successful reconstruction. 
Therefore, much of this work is focused on such statistical dimensions. 

Now we consider the case where the signal has several structures simultaneously. 
Two important examples are (i) simultaneously low rank and sparse matrices and (ii) tensors with several low rank matrizations. 
In such cases one often knows good regularizers for the individual structures. 
In this work, we address the question of how well one can use convex combinations of the individual reguralizers to regularize for the combined structure. 

\subsection{Related work}
It is a standard practice to account for multiple simultaneous structures in
convex optimization by combining different regularizers or constraints.
The hope is to 
improve the sampling rate, i.e., to recover $x_0$ from fewer
observations $y$.  Unfortunately, when taking simple combinations of
regularizers there are certain limitations on the improvement of sampling rate. 

For the prominent example of entrywise sparse and low-rank matrices Oymak et al.~\cite{Oymak2015} showed that convex combinations of $\ell_1$-- and nuclear norm cannot improve the scaling of the sampling rate $m$ in convex recovery. 
Mu et al.~\cite{MuHuaWri13} considered linear combinations of arbitrary norms and derived more explicit and simpler lower bounds on the sampling rate using elegant geometric arguments. 

In particular, this analysis also covers certain approaches to tensor recovery. 
It should be noted that low-rank tensor reconstruction using few measurements is notoriously difficult. 
Initially, it was suggested to use linear combinations of nuclear norms as a reguralizer \cite{GanRecYam11}, a setting where the lower bounds \cite{MuHuaWri13} apply. 
Therefore, the best guarantee for tractable reconstructions is still a non-optimal reduction to the matrix case \cite{MuHuaWri13}. 

If one gives up on having a convex reconstruction algorithm, non-convex quasi-norms can be minimized that lead to an almost optimal sampling rate \cite{GhaPlaYil17}. 
This reconstruction is for tensors of small canonical rank (a.k.a.\ CP rank). 
Also for this rank another natural regularizer might be the so-called tensor nuclear norm. 
This is another semi-norm for which one can find tractable semidefinite programming relaxations based on so-called $\theta$-bodies \cite{RauSto15}. 
These norms provide promising candidates for (complexity theoretically) efficient and guaranteed reconstructions. 

Also following this idea of atomic norm decompositions \cite{Chandrasekaran2010}, 
a single regularizer was found by Richard et al. \cite{RicOboVer14} that yields again optimal sampling rates at the price that the reconstruction is not give by a tractable convex program.

\subsection{Contributions and outline}
We discuss further the idea of taking convex combinations of regularizers from a convex analysis viewpoint. 
Moreover, we focus on the scenario where the weights in a maximum and also a sum of norms can depend on the unknown object $x_0$ itself, which reflects an optimal tuning of the convex program. 

Based on tools established in \cite{MuHuaWri13}, which may be not fully
recognized in the community, we discuss simple convex combinations of regularizers. 
As already pointed out by Oymak et al.~\cite{Oymak2015}, an optimally weighted
maximum of regularizers has the best sampling rate among such approaches. 
We follow this direction and discuss how sampling rates can be obtained in such setting. 
Specifically, we point out that the arguments leading to the lower bounds of Mu et al.~\cite{MuHuaWri13} can also be used to obtain generic lower bounds for the maximum of norms, which implies similar bounds for the sum of norms (Section~\ref{sec:generic:lowerbounds}). 
We also present an SDP characterization of dual norms for weighted sums and maxima and use them for numerically sampling the statistical dimension of descent cones, which correspond to the critical sampling rate \cite{AmeLotMcC14}. 

In Section~\ref{sec:sandl} we discuss the prominent case of simultaneously sparse and low rank matrices. 
Here, our contributions are twofold. 
We first show that the measurements satisfy a restricted isometry property (RIP) at a sampling rate that is essentially optimal for low rank matrices, which implies injectivity of the measurement map $\A$. 
Then, second, we provide numerical results showing that maxima of regularizers lead to recovery results that show an improvement over those obtained by the optimally weighted sum of norms above an intermediate sparsity level. 

Then, in Section~\ref{sec:tensors} we discuss the regularization for rank-$1$ tensors using combinations of matrix norms (extending the ``square deal'' \cite{MuHuaWri13} idea). 
In particular, we suggest to consider the maximum of several nuclear norms of ``balanced'' or ``square deal'' matrizations for the reconstruction of tensor products. 
We point out that the maximum of regularizers can lead to an improved recovery, often better
than expected. 

It is the hope that this work will contribute to more a comprehensive understanding of convex approaches to simultaneous structures in important recovery problems. 

%%%=============================================
\section{Lower bounds for convex recovery with combined regularizers}
%%% =============================================
\label{sec:generic:lowerbounds}
In this section we review the convex arguments used by Mu et al.\ \cite{MuHuaWri13} to establish lower bounds for the required number of measurements when using a linear combination of regularizers. 

\subsection{Setting and preliminaries}
For a positive integer $m$ we use the notation $[m]\coloneqq \{1,2,\dots,m\}$. 
The $\ell_p$-norm of a vector $x$ is denoted by $\norm{x}_{\ell_p}$ and Schatten $p$-norm of a matrix $X$ is denoted by $\norm{X}_p$. 
For $p=2$ these two norms coincide and are also called \emph{Frobenius norm} or \emph{Hilbert-Schmidt norm}. 
With slight abuse of notation, we generally denote the inner product norm of a Hilbert space by $\lTwoNorm{\argdot}$. 
For a cone $S$ and a vector $g$ in a vector space with norm $\norm{\argdot}$ we denote their induced \emph{distance} by   
\begin{equation}\label{eq:cone_distance}
\norm{g - S} \coloneqq \inf \{\norm{g-x} : x \in S   \} 
    \end{equation}
The \emph{polar} of a cone $K$ is 
\begin{equation}\label{eq:cone_dual}
K^\circ \coloneqq \{ y : \langle y, x \rangle \leq 0 \ \forall x \in K\} .
 \end{equation}
For a set $S$ we denote its convex hull by $\conv(S)$ and the cone generated by $S$ by $\cone(S) \coloneqq \{\tau x:\ x \in S, \ \tau >0 \}$. 
For convex sets $C_1$ and $C_2$ one has
  \begin{equation}\label{eq:cone_conv}
    \cone \conv(C_1 \cup C_2) = \cone(C_1)+\cone(C_2) \, ,
  \end{equation}
  where ``$\subset$'' follows trivially from the definitions. 
  In order to see ``$\supset$'' we write a conic combination $z=\rho x + \sigma y \in \cone(C_1)+\cone(C_2)$ as $z = (\rho + \sigma)\left( \frac{\rho}{\rho+\sigma} x + \left(1- \frac{\rho}{\rho+\sigma}\right) y \right)$.
  The \emph{Minkowski sum} of two sets $C_1$ and $C_2$ is denoted by $C_1+C_2 \coloneqq\{x+y: \, x\in C_1, \, y\in C_2\}$. 
  It holds that
  \begin{equation} \label{eq:conv(MinkowskiSum)}
    \cone(C_1 + C_2) 
    \subseteq 
    \cone(C_1) + \cone(C_2) \, .
  \end{equation}
The \emph{subdifferential} of a convex function $f$ at base point $x$ is denoted by $\partial f(x)$. 
  When $f$ is a norm, $f = \norm{\argdot}$, then the subdifferential is the set of vectors where H{\"o}lder's inequality is tight, 
  \begin{equation} \label{eq:def:subdifferential}
    \partial \norm{\argdot}(x) 
    = 
    \{ y: \ \langle y,x \rangle = \norm{x} , \ \norm{y}^\circ \leq 1 \} \, ,
  \end{equation}
  where $\norm{y}^\circ$ is the dual norm of $\norm{\argdot}$ defined by 
  $\norm{y}^\circ \coloneqq \sup_{\norm{x}=1} |\langle x,y\rangle|$. 
  The descent cone of a convex function $f$ at point $x$ is \cite[Definition~2.4]{Tro14}
\begin{equation}
  \DC(f,x) \coloneqq \cone \{y: f(x+y) \leq f(x)\} \, 
\end{equation}
and it holds that \cite[Fact~3.3]{Tro14}
\begin{equation}\label{eq:polarDC}
  \DC(f,x)^\circ = \cl \cone \partial f(x) \, , 
\end{equation}
where $\cl S$ denotes the closure of a set $S$. 
Let $\{f_i\}_{i\in[k]}$ be proper convex functions such that the relative interiors of their domains have at least a point in common.  
Then
  \begin{equation}\label{eq:subdiff:sum}
    \partial \left(\sum_{i\in[k]}\lambda_i f_i\right)(x) 
    =
    \sum_{i\in [k]}\lambda_i \partial f_i(x)  \, .
  \end{equation}
A function $f \coloneqq \max_{i\in [k]} f_i$ that
  is the point-wise maximum of convex functions $\{f_i\}_{i\in[m]}$ has the
  subdifferential \cite[Corollary~D.4.3.2]{HirLem04}
  \begin{equation}\label{eq:subdiff:max}
    \partial f(x) 
    = 
    \conv\left(\bigcup_{i: f_i(x) = f(x)} \partial f_i(x)\right) \, .
  \end{equation}
 Hence, the generated cone is the Minkowski sum
 \begin{equation}\label{eq:cone_subdiff:max}
    \cone \partial f(x) 
    = 
    \sum_{i: f_i(x) = f(x)} \cone \partial f_i(x) \, . 
 \end{equation}
The Lipschitz constant of a function $f$ 
w.r.t.\ to a norm $\norm{\argdot}$  
is the smallest constant $L$ such that for all vectors $x$ and $x'$
\begin{equation}\label{eq:lipschitz}
    |f(x)-f(x')|\leq L \, \norm{x-x'} \, .
  \end{equation}
Usually, $\norm{\argdot}$ is an $\ell_2$-norm, which also fits the Euclidean geometry of the circular cones defined below. 

The \emph{statistical dimension} of a convex cone $K$ is given by 
(see, e.g., \cite[Proposition~3.1(4)]{AmeLotMcC14})
\begin{equation}\label{eq:def:stat_dim}
  \delta(K)
  \coloneqq 
  \EE_g\left[ \fNorm{g-K^\circ}^2 \right]
 \end{equation}
 where $g \sim N(0,\1)$ is a Gaussian vector. 
The statistical dimension satisfies \cite[Proposition~3.1(8)]{AmeLotMcC14}
\begin{equation}\label{eq:stad_dim_polar}
  \delta(K) + \delta(K^\circ) = d 
\end{equation}
for any cone $K \subset V$ in a vector space $V\cong \RR^d$. 
Now, let us consider a compressed sensing problem where we wish to recover a signal $x_0$ from $m$ fully Gaussian measurements using a convex regularizer $f$. 
For small $m$, a successful recovery is unlikely and for large $m$ the recovery works with overwhelming probability. 
Between these two regions of $m$ one typically observes a phase transition. 
This transition is centered at a value given by the statistical dimension of the descent cone $\delta(\DC(f,x_0))$ of $f$ at $x_0$ \cite{AmeLotMcC14}. 
Thanks to \eqref{eq:polarDC}, this dimension can be calculated via the subdifferential of $f$, 
\begin{equation}\label{eq:stad_dim_as_dist}
  \delta(\DC(f,x_0))
  =
  \EE_g\left[ \fNorm{g-\cone \partial f(x_0)}^2 \right] \, .
\end{equation}
By $\gtrsim$ and $\lesssim$ we denote the usual greater or equal and less or equal relations up to a positive constant prefactor and $\propto$ if both holds simultaneously. 
Then we can summarize that the convex reconstruction \eqref{eq:rec_def} requires exactly a number of measurements $m \gtrsim \delta(\DC(f,x_0))$, which can be calculated via the last equation.

%%%=============================================
\subsection{Combined regularizers} \label{subsec:combined_regularizers}
%%% =============================================
We consider the convex combination and weighted maxima of norms $\norm{\argdot}_{(i)}$,
where $i = 1, 2, \dots, k$ and set
\begin{equation}\label{eq:def_max_sum_reg_norm}
  \begin{aligned}
    \regNormMax{\argdot} &\coloneqq \max_{i\in[k]} \mu_{i} \norm{\argdot}_{(i)}
    \\
    \regNormSum{\argdot} &\coloneqq \sum_{i\in[k]} \lambda_{i} \norm{\argdot}_{(i)} \, , 
  \end{aligned}
\end{equation}
where $\lambda=(\lambda_1,\dots,\lambda_k)\geq 0$ and $\mu=(\mu_1,\dots,\mu_{k})\geq0$ are to be chosen later. 
Here, we assume that we have access to single norms of our original signal $\norm{x_0}_{(i)}$ so that we can  fine-tune the parameters $\lambda$ and $\mu$ accordingly. 

In \cite{Oymak2015} lower bounds on the necessary number of
measurements for the reconstruction \eqref{eq:rec_def} using general
convex relaxations were derived.
However, it has not been worked out explicitly what can be said in the
case where one is allowed to choose the weights $\lambda$ and $\mu$ 
dependending on the signal $x_0$.  For the sum norm
$\regNormSum{\argdot}$ this case is covered by the lower bounds in
\cite{MuHuaWri13}.  Here, we will see that the arguments extend to
optimally weighted max norms, i.e., $\regNormMax{\argdot}$ with $\mu$
being optimally chosen for a given signal.

A description of the norms dual to those given by 
\eqref{eq:def_max_sum_reg_norm} is provided in Section~\ref{sec:dual_norms}.

\subsection{Lower bounds on the statistical dimension}
\label{subsec:generic:lowerbounds}
The statistical dimension of the descent cone is obtained as Gaussian squared distance from the cone generated by the subdifferential \eqref{eq:stad_dim_as_dist}. 
Hence, it can be lower bounded by showing (i) that the subdifferential is contained in a small cone and (ii) bounding that cone \cite{MuHuaWri13}. 
A suitable choice for this small cone is the so-called circular cone.

\subsubsection{Subdifferentials and circular cones}
We use the following statements from \cite{AmeLotMcC14} and
\cite[Section~3]{MuHuaWri13} which show that subdifferentials are
contained in circular cones.  More precisely, we say that the angle
between vectors $x,z\in \RR^\totdim$ is $\theta$ if
$\langle z,x\rangle = \cos(\theta)\, \lTwoNorm{z} \lTwoNorm{x}$ and
write in that case $\angle(x,z) = \theta$.  Then the \emph{circular
  cone} with axis $x$ and angle $\theta$ is defined as
\begin{align} \label{eq:def:circ_cone}
  \circone(x,\theta) 
  &\coloneqq 
  \{ z: \ \angle(x,z) \leq \theta \} \, .
\end{align}
  Its statistical
dimension has a simple upper bound in terms of its dimension: For all
$x\in V \cong \RR^\totdim$ and some $\theta \in [0,\pi/2]$
\cite[Lemma~5]{MuHuaWri13}
\begin{equation}\label{eq:stat_dim_circ_cone}
  \delta(\circone(x,\theta))\leq \totdim \sin(\theta)^2 +2 \, .
\end{equation}
By the following argument this bound can be turned into a lower bound on the statistical dimension of descent cones and, hence, to the necessary number of measurements for signal reconstructions. 
The following two propositions summarize arguments made by Mu et al.~\cite{MuHuaWri13}.
\begin{proposition}[Lower bound on descent cone statistical dimensions \cite{MuHuaWri13}]
\label{prop:deltaDC_lower_bound}
Let us consider a convex function $f$ as a regularizer for the recovery of a signal $0\neq x_0 \in V \cong \RR^d$ in a $d$-dimensional space. 
If $\partial f(x_0) \subset \circone(x_0,\theta)$ then 
\begin{equation}
  \delta(\DC(f,x_0)) \geq d \cos^2\theta - 2 \, .
\end{equation}
\end{proposition}

This statement is already implicitly contained in \cite{MuHuaWri13}. 
The arguments can be compactly summarized as follows. 

\begin{proof}
With the polar of the descent cone \eqref{eq:polarDC}, the assumption, and the statistical dimension of the polar cone \eqref{eq:stad_dim_polar} we obtain 
\begin{align}
  \delta(\DC(f,x_0))
  &=
  d - \delta(\cone \partial f(x_0))
  \\
  &\geq
  d - \delta(\circone(x_0,\theta))
\end{align}
Hence, with the bound on the statistical dimension of the circular cone \eqref{eq:stat_dim_circ_cone}, 
\begin{equation}
  \delta(\DC(\regNormMax[\mu^\ast\!]{\argdot},x_0))
  \geq
  d (1-\sin^2\theta)-2 = d \cos^2\theta - 2\, .
\end{equation}
\qed\end{proof}

Recall from \eqref{eq:lipschitz} a norm $f$ is $L$-Lipschitz
with respect to the $\ell_2$--norm 
on a (sub-)space $V$ 
if
\begin{equation}
  f(x)\leq L \|x\|_{\ell_2} 
  \label{eq:lipschitz:norm}
\end{equation}
for all $x\in V$. 
Note that this implies for the dual norm $f^\circ$ that
\begin{equation}\label{eq:lipschitz:dnorm}
  \|x\|_{\ell_2}\leq Lf^\circ(x) 
\end{equation}
for all $x\in V$.

\begin{proposition}[{\cite[Section~3]{MuHuaWri13}}]
\label{prop:delta_lower_bound_of_cos_theta}
  Let $f$ be a norm on $V \cong \RR^\totdim$ that is $L$-Lipschitz
  with respect to the $\ell_2$-norm on $V$.
  Then 
  \begin{align}
    \partial f(x_0)
    &\subseteq
      \left\{ x:\ 
      \frac{\braket{x}{x_0}}{\lTwoNorm{x}\lTwoNorm{x_0}} 
      \geq 
      \frac{f(x_0)}{L \lTwoNorm{x_0}} 
      \right\}
      = 
      \circone(x_0, \theta)
  \end{align}
  with $\cos(\theta) = \frac{f(x_0)}{L \lTwoNorm{x_0}}$. 
  \label{prop:generic:subdiff:circone}
\end{proposition}
For sake of self-containedness we summarize the proof of \cite[Section~3]{MuHuaWri13}. 

\begin{proof}
  The subdifferential of a norm $f$ on $V \cong \RR^\totdim$ with dual norm $f^\circ$ is given by 
  \begin{equation}
    \partial f(x_0) 
    = 
    \{ x:\ \braket {x}{x_0} =  f(x_0) , \ f^\circ(x) \leq 1 \} \, .
    \label{eq:prop:generic:subdiff:proof:1}
  \end{equation}
  Then, thanks to H\"older's inequality
  $\braket{x}{x_0} \leq f^\circ(x) f(x_0)$, 
  we have for any subgradient $x \in \partial f(x_0)$ 
  \begin{equation}
    \frac{\braket{x}{x_0}}{\lTwoNorm{x}\lTwoNorm{x_0}}
    \overset{\eqref{eq:prop:generic:subdiff:proof:1}}{=}
    \frac{f(x_0)}{\lTwoNorm{x_0}\lTwoNorm{x}}
    \overset{\eqref{eq:lipschitz:dnorm}}{\geq}
    \frac{f(x_0)}{L \lTwoNorm{x_0}f^\circ(x)}
    \overset{\eqref{eq:prop:generic:subdiff:proof:1}}{\geq}
    \frac{f(x_0)}{L \lTwoNorm{x_0}}
  \end{equation} 
  This bound directly implies the claim. 
\qed\end{proof}

So together, Propositions~\ref{prop:deltaDC_lower_bound} and~\ref{prop:delta_lower_bound_of_cos_theta} imply 
\begin{equation}
  \delta(\DC(f,x_0)) \geq \frac{d}{L^2} \rank_{f}(x_0)-2 \, ,
\end{equation}
where 
\begin{equation}
  \rank_{f}(x_0) \coloneqq \frac{f(x_0)^2}{\lTwoNorm{x_0}}
\end{equation}
is the \emph{$f$-rank} of $x_0$ (e.g., effective sparsity for $f = \lOneNorm{\argdot}$ and effective matrix rank for $f= \OneNorm{\argdot}$).

\subsubsection{Weighted regularizers}
A larger subdifferential leads to a smaller statistical dimension of
the descent cone and,
hence, to a reconstruction with a potentially smaller number of
measurements in the optimization problems
\begin{equation}
  \min \regNormMax{x} \quad \text{subject to } \A(x) = y
  \label{eq:generic:max:cvx}
\end{equation}
and 
\begin{equation}
  \min \regNormSum{x}\quad\text{subject to } \A(x) = y
  \label{eq:generic:sum:cvx}
\end{equation}
with the norms \eqref{eq:def_max_sum_reg_norm}. 
Having the statistical dimension \eqref{eq:stad_dim_as_dist} in mind, 
we set
\begin{align}
  \deltasum(\lambda)
  &\coloneqq 
  \EE_g\left[ \fNorm{g-\cone \partial \regNormSum{\argdot}(x_0)}^2 \right]
  \\
  \deltamax(\mu)
  &\coloneqq 
  \EE_g\left[ \fNorm{g-\cone \partial \regNormMax{\argdot}(x_0)}^2 \right]
  \, , 
\end{align}
which give the optimal number of measurements in a precise sense \cite[Theorem~II]{AmeLotMcC14}. 
Note that it is clear that $\deltasum(\lambda)$ is continuous in $\lambda$.

Now we fix $x_0$ and consider adjusting the weights $\lambda_i$ and
$\mu_i$ depending on $x_0$. 
We will see the geometric arguments from \cite{MuHuaWri13} extend to that case. 
Oymak et al.~\cite[Lemma~1]{Oymak2015} show that the max-norm $\regNormMax{\argdot}$ with weights $\mu_i$ chosen as 
\begin{equation}\label{eq:optimal_weights}
  \mu_i^\ast \coloneqq \frac{1}{\norm{x_0}_{(i)}} 
\end{equation}
has a better recovery performance than all other convex combinations of norms. 
With this choice the terms in the maximum \eqref{eq:def_max_sum_reg_norm} defining $\regNormMax{\argdot}$ are all equal. 
Hence, from the subdifferential of a maximum of functions \eqref{eq:subdiff:max} it follows that this choice of weights is indeed optimal. 
Moreover, the optimally weighted max-norm $\regNormMax[\mu^\ast]{\argdot}$ leads to a better (smaller) statistical dimension for $x_0$ than all sum-norms $\regNormSum{\argdot}$ and, hence, indeed to a better reconstruction performance:

\begin{proposition}[Optimally weighted max-norm is better than any sum-norm]\label{prop:opt_max_is_best}
Consider a signal $x_0$ and the corresponding optimal weights $\mu^\ast$ from \eqref{eq:optimal_weights}. 
Then, 
for all possible weights $\lambda\succeq 0$ in the sum-norm, 
we have
\begin{equation}
  \cone \partial \regNormSum{\argdot}(x_0) 
  \subset
  \cone \partial \regNormMax[\mu^\ast]{\argdot}(x_0) \, .
\end{equation}
In particular, $\deltamax(\mu^\ast) \leq \deltasum(\lambda)$ for all $\lambda\succeq 0$.  
\end{proposition}

\begin{proof}
Using \eqref{eq:subdiff:max} and \eqref{eq:cone_conv} we obtain 
\begin{align}
\cone \partial \regNormMax{\argdot}(x_0) 
&= 
\cone \conv\left(\bigcup_{i\in [k]} \partial \norm{\argdot}_{(i)}(x_0)\right)
\\
&=\sum_{i\in [k]} \cone \partial \norm{\argdot}_{(i)}(x_0) 
\label{eq:cone_max_norm}
\end{align}
with \eqref{eq:subdiff:sum} 
\begin{equation}
  \partial \regNormSum{\argdot}(x_0) 
  =
  \sum_{i \in [k]} \lambda_i \partial\norm{\argdot}_{(i)}(x_0) \, .
\end{equation}
Then \eqref{eq:conv(MinkowskiSum)} concludes the proof. 
\qed\end{proof}

Proposition~\ref{prop:opt_max_is_best} implies that if the max-norm minimization \eqref{eq:generic:max:cvx} with optimal weight \eqref{eq:optimal_weights} does not recover $x_0$, then the sum-norm minimization \eqref{eq:generic:sum:cvx} can also not recover $x_0$ for any weight $\lambda$.  

Now we will see that lower bounds on the number of measurements from \cite[Section~3]{MuHuaWri13} straightforwardly extend to the max-norm with optimal weights. 
These lower bound are obtained by deriving an inclusion into a circular cone. 
Combining Proposition~\ref{prop:delta_lower_bound_of_cos_theta} with \eqref{eq:cone_max_norm} and noting that a Minkowski sum of circular cones \eqref{eq:def:circ_cone} is just the largest circular cone yields the following:

\begin{proposition}[Subdifferential contained in a circular cone]
\label{prop:subdiff_circcone}
Let $0\neq x_0\in V \cong \RR^\totdim$ (signal) and $\norm{\argdot}_{(i)}$ be norms satisfying $\norm{x}_{(i)}\leq L_i \lTwoNorm{x}$ for $i\in [k]$ and for all $x\in V$. 
Then the subdifferential of 
\begin{equation}\label{eq:opt_weighted_Max_norm}
  \regNormMax[\mu^\ast\!]{\argdot}
  \coloneqq 
  \max_{i\in[k]} \, \frac{\norm{\argdot}_{(i)}}{\norm{x_0}_{(i)}} 
\end{equation}
satisfies
\begin{equation}
  \partial \regNormMax[\mu^\ast\!]{\argdot} (x_0)
  \subset \circone(x_0,\theta)
% \end{equation}
\quad 
\text{with}
\quad
% \begin{equation}
  \cos(\theta) 
  = 
  \max_{i\in[k]} \frac{\norm{x_0}_{(i)}}{L_i \lTwoNorm{x_0}} \, .
  \label{eq:costheta}
\end{equation}
\end{proposition}

In conjunction with Proposition~\ref{prop:deltaDC_lower_bound} this yields the bound
\begin{equation} \label{eq:deltaDC_bound}
  \delta(\DC(\regNormMax[\mu^\ast\!]{\argdot},x_0))
  \geq
  \max_{i\in[k]} \frac{d\norm{x_0}_{(i)}^2}{L_i^2 \lTwoNorm{x_0}^2} -2 \, .
\end{equation}

Hence, upper bounds on the Lipschitz constants of the single norms
yield a circular cone containing the subdifferential of the maximum,
where the smaller the largest upper bound the smaller the circular
cone. 
In terms of $f$-ranks it means that 
\begin{equation} 
  \delta(\DC(\regNormMax[\mu^\ast\!]{\argdot},x_0))
  \geq
  \max_{i\in[k]} \frac{d}{L_i^2}\rank_{\norm{\argdot}_{(i)}}(x_0) - 2 \,  .
\end{equation}

Now, \cite[Lemma~4]{MuHuaWri13} can be replaced by this slightly more general proposition and the main lower bound \cite[Theorem~5]{MuHuaWri13} on the number of measurements $m$ still holds. 
The factor $16$ in \cite[Corollary~4]{MuHuaWri13} can be replaced by an $8$ (see updated Arxiv version \cite{AmeLotCoy13arxiv} of \cite{AmeLotMcC14}). 
These arguments 
(specifically, \cite[(7.1) and (6.1)]{AmeLotCoy13arxiv} with $\lambda \coloneqq \delta(C)-m$)
show the following for the statistical dimension $\delta$ of the descent cone. 
The probability of successful recovery for $m\leq \delta$ is upper bounded by $p$ for all $m \leq \delta - \sqrt{8 \delta \ln(4/p)}$.
Conversely, the probability of unsuccessful recovery for $m\geq \delta$ is upper bounded by $q$ for any $m \geq \delta + \sqrt{8 m \ln(4/q)}$, 
so in particular, for any $m \geq \delta + \sqrt{8 d \ln(4/q)}$. 
Moreover, this yields the following implication of \cite{AmeLotMcC14}: 

\begin{theorem}[Lower bound]
\label{thm:generic:lowerbound:max}
  Let $\norm{\argdot}_{(i)}$ be norms satisfying
  $\norm{x}_{(i)}\leq L_i \lTwoNorm{x}$ for $i\in [k]$ and for all $x\in V$. 
  Fix $0\neq x_0\in V \cong \RR^\totdim$ (signal) and set
  \begin{equation} \label{eq:min_nr_of_measurements}
    \kappa
    \coloneqq 
    \min_{i\in [K]} \frac{\totdim \norm{x_0}_{(i)}^2}{L_i^2 \lTwoNorm{x_0}^2} -2 
    \quad \text{(min. number of measurements)}.
  \end{equation}
  Consider the optimally weighted max-norm \eqref{eq:opt_weighted_Max_norm}. 
  Then, for all $m\leq \kappa$, the probability $\psuccmax$ that
  $x_0$ is the unique minimizer of the reconstruction program \eqref{eq:generic:max:cvx} is bounded by
  \begin{equation}\label{eq:generic:lowerbound:max}
    \psuccmax 
    \leq 
    4 \exp\left(-\frac{(\kappa-m)^2}{8\kappa} \right) \, .
  \end{equation}
\end{theorem}

We will specify the results in more detail for the sparse and low-rank
case in Corollary \ref{cor:sl:lowerbound:max} and an example for the tensor case in Eq.~\eqref{eq:lower_bound_funny3}.

%%% ----------------------------------------------------
\subsection{Estimating the statistical dimension via sampling and semidefinite programming}
%%% ----------------------------------------------------
Often one can characterize the subdifferential of the regularizer by a semidefinite program (SDP). 
In this case, one can estimate the statistical dimension via sampling and solving such SDPs. 

In more detail, in order to estimate the statistical dimension \eqref{eq:stad_dim_as_dist} for a norm $f$, 
we sample the Gaussian vector $g \sim \mc N(0,\1)$, 
calculate the distance $\fNorm{g-\cone \partial f(x_0)}^2$ using the SDP-formulation of $\partial f(x_0)$ and take the sample average in the end. 
In order to do so, we wish to also have an SDP characterization of the dual norm $\partial f^\circ$ of $f$, 
which provides a characterization of the subdifferential \eqref{eq:def:subdifferential} of $f$. 

In the case when the regularizer function $f$ is either $\regNormMax{\argdot}$ or $\regNormSum{\argdot}$, where the single norms $\norm{\argdot}_{(i)}$ have simple dual norms, 
we can indeed obtain such an SDP characterization of the dual norm $f^\circ$. 

\subsubsection{Dual norms}
\label{sec:dual_norms}
%
%In the case of the sum of the nuclear and $\ell_1$-norm this has been
%established already in \cite{Doan2013}.
It will be convenient to have explicit formulae for the dual norms to the combined regularizers 
defined in Section \ref{subsec:combined_regularizers}. 
\begin{lemma}[Dual of the maximum/sum of norms]\label{lem:dual_norms}
  Let $\norm{\argdot}_{(i)}$ be norms ($i\in [k]$) and denote by 
  $\regNormMax{\argdot}$ and 
  $\regNormSum{\argdot}$ 
  their weighted maximum and weighted sum as defined in \eqref{eq:def_max_sum_reg_norm}. 
  Then their dual norms are given by
  \begin{align}
  \regNormMax{y}^\circ &=  \inf\left\{\sum_{i=1}^k  \frac{1}{\mu_i} \norm{x_i}_{(i)}^\circ : \ y = \sum_{i=1}^k x_i \right\} 
  \label{eq:polar_of_max}
  \\
  \regNormSum{y}^\circ &=    \inf\left\{\max_{i\in[k]} \frac{1}{\lambda_i} \norm{x_i}_{(i)}^\circ : \ y = \sum_{i=1}^k x_i \right\} 
  \label{eq:polar_of_sum}
  .
  \end{align}
\end{lemma}

Statements of similar nature are well-known in functional analysis as well as in 
classical convex geometry (in the language of support functions and polar sets)
or in convex analysis (in the more general context of lower semi-continuous 
convex functions and their Legendre-Fenchel transforms). For completeness, we 
will provide a sketch of the argument. It will be convenient to 
use the notation from the convex analysis book by Rockafellar \cite{Rockafellar1996}. 
If $C\subset \RR^d$,  
 its \emph{support function} is defined by
\begin{equation}
  \delta^\ast_C(x) \coloneqq \sup_{y\in C} \langle x, y \rangle . 
\end{equation}
%and its indicator function by 
%\begin{equation}
%  \delta_C(x) \coloneqq 
%  \begin{cases} 0& \text{ if $x \in C$}
%    \\
%    \infty& \text{ otherwise.} 
%  \end{cases}
%\end{equation}
and the \emph{polar} of $C$ by
\begin{equation}
  C^\circ \coloneqq \{ y : \langle y, x \rangle \leq 1 \ \forall x \in C\}\, .
\end{equation}
(Note that while formally different, this definition is consistent with 
the polar of a cone introduced in \eqref{eq:cone_dual}.) 
If $C$ is closed convex and contains the origin, then we define its \emph{gauge function} by
\begin{equation}
\gamma_C(x) \coloneqq \inf \{ \lambda \geq 0 : x \in \lambda C \} . 
\end{equation}
%The polar of a gauge function $\gamma$ is denoted by 
%$\gamma^\circ(y) \coloneqq \inf\{\lambda \geq 0 : \langle y,x \rangle \leq \lambda \gamma(x) \}$ 
%and the convex conjugate of a function (Fenchel-Legendre dual) $f$ by $f^\ast$. 
%The infimal convolution of convex functions $f$ and $g$ is defined by
%$f\square g(x) \coloneqq \inf_y\{ f(x-y) + g(y) \}$. 

The archetypal context for the above notions is as follows. 
If $B$ is the unit ball with respect to some norm, i.e., $B=\{x : \norm{x}\leq 1\}$, 
then $\gamma_B=\norm{\argdot}$, while $\gamma_{B^\circ}=\delta_B^\ast$ 
coincide with the dual norm $\norm{\argdot}^\circ$. 
(In particular, $B^\circ$ is the unit ball with respect to $\norm{\argdot}^\circ$.) 

There are all sorts of elementary rules \cite{Rockafellar1996} that we will use. 
%Let $f$ and $g$ be lower semi-continuous convex functions and 
Let $C_1$, $C_2$ and $B$, $B_1$, $B_2$ be closed convex sets where $B, B_1,B_2$ contain the origin. 
Then
%$f^{\ast\ast} = f$, 
%$\gamma_B^{\circ\circ} = \gamma_B$, 
$\delta_{C_1}^\ast + \delta_{C_2}^\ast = \delta_{C_1+C_2}^\ast$, 
$B^{\circ\circ}= B$ (the bipolar theorem), 
%$(f\square g)^\ast = f^\ast + g^\ast$, 
and $\delta^\ast_B = \gamma_{B^\circ}$. % and $(\delta_C)^\ast = \delta_C^\ast$. 
 Next, if $B=B_1\cap B_2$,  then $\gamma_B = \max \{\gamma_{B_1},\gamma_{B_1}\}$ 
and $B^\circ= (B_1\cap B_2)^\circ = \conv(B_1^\circ \cup B_2^\circ)$. 

\begin{proof}[Proof of Lemma~\ref{lem:dual_norms}] 
  First, by rescaling the norms we can restrict our attention to the case when 
  all $\lambda_i$ and all $\mu_i$ are equal to $1$. 
  %For the sake of better readability we focus on the case $k=2$. 
  %The general case follows analogously. 
  Next, to reduce the clutter we will focus on the case $k=2$ 
  (the general case follows {\em mutatis mutandis}, or by induction) 
  and denote $B_i\coloneqq \{x: \ \norm{x}_{(i)}\leq 1\}$, $i=1,2$. 
  The argument actually works for general gauges, 
  in particular without the symmetry assumption $B_i = -B_i$.
  
  %In order to establish the statement of the lemma we will use the fact that 
  %the norms are the gauge functions of their unit balls, i.e., 
  %$\norm{\argdot}_{(i)} = \gamma_{B_i}$ with $B_i\coloneqq \{x: \ \norm{x}_{(i)}\leq 1\}$. 
  %For the proofs, however, let $B_i$ be closed convex sets containing the origin 
  %(they do not need to satisfy $B_i = -B_i$, as is the case for norms). 

  In order to establish the relevant case of \eqref{eq:polar_of_max} we start 
  with the identity
  \begin{equation}
   \max\{\norm{\argdot}_{(1)}, \norm{\argdot}_{(2)} \} =   
   \max\{ \gamma_{B_1}, \gamma_{B_2} \} =   \gamma_{B_1 \cap B_2}.
  \end{equation}
  Accordingly, the unit ball in the dual norm is  
  $(B_1 \cap B_2)^\circ = \conv(B_1^\circ \cup B_2^\circ)$. 
  In other words, the dual norm of $y$ is at most $1$ iff $y = (1-t)y_1+ty_2$
  for some $t\in [0,1]$ and  $y_i \in B_i^\circ$, $i=1,2$. Denoting $x_1 = (1-t)y_1$ 
  and $x_2=ty_2$, this translates to  $y = x_1+x_2$  and 
  $\norm{x_1}_{(1)}^\circ + \norm{x_2}_{(2)}^\circ \leq 1$. 
  So  the left hand side of \eqref{eq:polar_of_max} is at most $1$ iff 
   the right hand side is, and the general case follows by homogeneity of the norm. 
   
  The case of \eqref{eq:polar_of_sum} is even simpler. We have 
  \begin{equation}
  \norm{\argdot}_{(1)} + \norm{\argdot}_{(2)}=\gamma_{B_1}+ \gamma_{B_2}
  =\delta_{B_1^\circ}^\ast+ \delta_{B_2^\circ}^\ast 
  = \delta_{B_1^\circ+B_2^\circ}^\ast= \gamma_{(B_1^\circ+B_2^\circ)^\circ}. 
  \end{equation}
    While the ``polar body $(K+L)^\circ$ of a [Minkowski] 
  sum of convex bodies has no plausible interpretation in terms of 
  $K^\circ$, $L^\circ$'' \cite{stackexchange:polar}, 
  the bipolar theorem tells us that the unit ball of the dual norm in question is 
  $B_1^\circ+B_2^\circ$. In other words, the dual norm of $y$ is at most $1$ iff 
   $y = x_1+x_2$ for some $x_i \in B_i^\circ$, $i=1,2$, i.e., 
   verifying $\max\{\norm{x_1}_{(1)}^\circ, \norm{x_2}_{(2)}^\circ\} \leq 1$, 
   and we conclude as before. 
  % \begin{split}
    %  (\gamma_{B_1} + \gamma_{B_2})^\circ 
     % &=
     %(\delta^\ast_{B_1^\circ} + \delta^\ast_{B_2^\circ})^\circ    
     % = 
     % (\delta^\ast_{B_1^\circ +B_2^\circ})^\circ  \\
     % &=
     % \bigl(\gamma_{(B_1^\circ+B_2^\circ)^\circ}\bigr)^\circ
     % =
     % \gamma_{B_1^\circ+B_2^\circ}\, ,
   % \end{split}
\end{proof}

\subsubsection{Gaussian distance as SDP}
\label{subsec:generic:gaussiandist:SDP}
We were aiming to estimate the statistical dimension \eqref{eq:stad_dim_as_dist} by sampling over SDP outcomes over Gaussian vectors $g$.  
For any vector $g$ the distance to the cone generated by the subdifferential of a norm $f$ is
\begin{align}
\fNorm{g - \cone \partial f(x_0)}
&=
\min\left\{ \fNorm{g - \tau x}: \ 
  \tau\geq 0, \ 
  \langle x, x_0 \rangle = f(x_0) , \
  f^\ast(x) \leq 1
  \right\}
  \\
  &=
\min\left\{ \fNorm{g - x}: \ 
  \tau\geq 0, \ 
  \langle x, x_0 \rangle = \tau f(x_0) , \
  f^\ast(x) \leq \tau
  \right\}
\end{align}
For $f = \regNormSum{\argdot}$ we use its polar \eqref{eq:polar_of_sum} to obtain
\begin{align}
&\phantom{=.} \nonumber
\fNorm{g - \cone \partial \regNormSum{\argdot}(x_0)}
\\
&=
\min_{\tau,x}\left\{ \fNorm{g - \tau x}: 
  \tau\geq 0, \ 
  \langle x, x_0 \rangle = \regNormSum{x_0} , \
  \inf_{\{x_i\}_{i=1}^{[k]}}\left\{ \max_{i\in [k]} \frac{\norm{x_i}_{(i)}^\circ}{\lambda_i}: \ x = \sum_{i=1}^k x_i \right\} \leq 1
  \right\} \label{eq:dist_inf_inf}
  \\
&=
\min_{\tau,x,\{x_i\}_{i=1}^{[k]}}\left\{ \fNorm{g - \tau x}: \ 
  \tau \geq 0, \ 
  \langle x, x_0 \rangle = \regNormSum{x_0} , \
  \frac{\norm{x_i}_{(i)}^\circ}{\lambda_i} \leq 1 , \
  x = \sum_{i=1}^k x_i
  \right\} ,
  \label{eq:dist_inf}
\end{align}
where one needs to note that an optimal feasible point of \eqref{eq:dist_inf_inf} also yields an optimal feasible point of \eqref{eq:dist_inf} and vice versa. 
But this implies that 
\begin{equation}
  \begin{split}
    &\fNorm{g - \cone \partial \regNormSum{\argdot}(x_0)}\\
    &=
    \min_{\tau,\{x_i\}_{i=1}^{[k]}}
    \left\{
      \fNorm{g - \sum_i x_i}: \ 
      \sum_i \langle x_i, x_0 \rangle = \tau \regNormSum{x_0} , \
      \norm{x_i}_{(i)}^\circ \leq \tau \lambda_i, \
      \tau \geq 0
    \right\} \, .
    \label{eq:generic:Gaussian_dist_sum_SDP}
  \end{split}
\end{equation}

For the maximum of norms regularizer we again choose the optimal weights \eqref{eq:optimal_weights} to ensure that all norms
are active, i.e.,
\begin{equation}
  \mu^\ast_i \norm{x_0}_{(i)} = \regNormMax[\mu^\ast\!]{x_0} \, .
\end{equation}
Then we use the subdifferential expression \eqref{eq:cone_subdiff:max} for a point-wise maximum of functions to obtain
\begin{equation}
  \begin{split}
    &\phantom{={}}\fNorm{g - \cone \partial \regNormMax[\mu^\ast\!]{\argdot}(x_0)}
    =
    \fNorm{g - \sum_{i=1}^k \cone \partial \norm{\argdot}_{(i)}(x_0)}
    \\
    &=
    \min_{\{\tau\}_{i=1}^{[k]},\{x_i\}_{i=1}^{[k]}}
    \left\{
      \fNorm{g - \sum_i x_i}: \ 
      \langle x_i, x_0 \rangle = \tau_i \regNormMax[\mu^\ast\!]{x_0} , \
      \norm{x_i}_{(i)}^\circ \leq \tau_i \mu^\ast_i, \
      \tau_i \geq 0
    \right\} \, .
    \label{eq:generic:Gaussian_dist_max_SDP}
  \end{split}
\end{equation}
In the case that $\norm{x_i}_{(i)}^\circ$ are $\ell_\infty$-norms or spectral norms these programs can we written as SDPs and be solved by standard SDP solvers.

%%% =================================================
\section{Simultaneously sparse and low-rank matrices}
%%% =================================================
\label{sec:sandl}
A class of structured signals that is important in many applications
are matrices which are simultaneously sparse and of low rank.  Such
matrices occur in sparse phase retrieval%
\footnote{See the cited works for further references on the classical
  non-sparse phase retrieval problem} \cite{Oymak2015,Li2013,
  Jaganathan2013}, dictionary learning and sparse encoding
\cite{Rubinstein2010}, sparse matrix approximation \cite{Smola2000},
sparse PCA \cite{Johnstone2009}, bilinear compressed sensing problems
like sparse blind deconvolution
\cite{Lee2013,Lee2015,Flinth2016,Lee:sampta15,Aghasi2017,Geppert2018}
or, more general, sparse self-calibration \cite{Ling2015}. For example,
upcoming challenges in communication engineering and signal processing
require efficient algorithms for such problems with theoretical
guarantees
\cite{Jung2014,WunBocStr15,Wunder2015:sparse5G,RotKliWun17}.  It is
also well-known that recovery problems related to simultaneous
structures like sparsity and low-rankness are at optimal rate often as
hard as the classical planted/ hidden clique problems, see for
example \cite{Berthet2013} for further details and references.

%%% ---------------------------------------
\subsection{Setting}
%%% ---------------------------------------
We consider the
$d=n_1\cdot n_2$--dimensional vector space $V=\KK^{n_1 \times n_2}$ of
$n_1\times n_2$-matrices with components in the field $\KK$ being either $\RR$ or $\CC$. 
The space $V$ is equipped with the Hilbert-Schmidt inner product 
defined by
% and induced norm 
\begin{equation}
  \langle X,Y\rangle\coloneqq \Tr (X^\ad Y) \, .
  % \quad\text{and}\quad
  % \TwoNorm{X} = \sqrt{\langle X,X\rangle} \, .
\end{equation}
Our core problem is to recover structured matrices from $m$ linear
measurements of the form $y=\A(X)$ with a linear map
$\A : \KK^{n_1 \times n_2} \to \KK^m$. Hence, a single measurement is
\begin{equation}
  y_i=\A(X)_i=\langle A_i,X\rangle \, .
  \label{eq:sandl:meas}
\end{equation}
It is clear that without further a-priori assumptions on the unknown
matrix $X$ we need $m\geq d=n_1\cdot n_2$ measurements to be able to successfully reconstruct $X$. 

As additional structure, we consider subsets of $V$ containing simultaneously low-rank and sparse matrices. 
However, there are
different ways of combining low-rankness and sparsity. For example one
could take matrices of rank $r$ with different column and row
sparsity, i.e., meaning that there are only a small number of non-zero
rows and columns. 
Here, we consider a set having even more structure which is motivated
by atomic
decomposition \cite{Chandrasekaran2010}.  Recall, that
the rank of a matrix $X$ can be defined as its ``shortest
description'' as
\begin{equation}
  \rank(X)
  \coloneqq
  \min\{r\,:\,X=\sum_{i=1}^r{x_iy_i^\ad}\,,\, x_i\in\mathbb{R}^{n_1}\,,\,y_i\in\mathbb{R}^{n_2}\} \, .
  \label{eq:std:rank}
\end{equation}
This characterization giving rise to the nuclear norm $\OneNorm{\argdot}$ as the corresponding atomic norm, see \cite{Chandrasekaran2010} for a nice introduction
to inverse problems from this viewpoint.  
Restricting the sets of feasible vectors $\{x_i\}$ and $\{y_i\}$ yields
alternative formulations of rank. In the case of sparsity, one can
formally ask for a description in terms of $(s_1,s_2)$-sparse atoms:
\begin{equation}
  \rank_{s_1,s_2}(X)
  \coloneqq
  \min\{r\,:\,X=\sum_{i=1}^r{x_iy_i^\ad}\,,\,\norm{x_i}_{\ell_0}\leq
  s_1,\norm{y_i}_{\ell_0}\leq s_2\} \, ,
\end{equation}
where $\norm{x}_{\ell_0}$ denotes the number of non-zero elements of a vector $x$. 
The idea of the corresponding atomic norm \cite{Chandrasekaran2010} has been worked for the
sparse setting by Richard et al.~\cite{RicOboVer14}.
Note that, compared to \eqref{eq:std:rank},
we do not require that $\{x_i\}_{i=1}^r$ and $\{y_i\}_{i=1}^r$ are orthogonal.
  
We say that a matrix $X\in \KK^{n_1 \times n_2}$ is simultaneously
\emph{$(s_1,s_2)$-sparse and of rank $r$} if it is in the set
\begin{equation}
  \M^{n_1\times n_2}_{r,s_1,s_2} \coloneqq
  \Bigl\{ \sum_{i=1}^r x_i y_i^\ad: \ 
  \norm{x_i}_0 \leq s_1, \ \norm{y_i}_0 \leq s_2 \Bigr\} .
  \label{eq:sandl:Mdef}
\end{equation}
Our model differs to the joint-sparse
  setting as used in \cite{Lee2013}, since the vectors
  $\{x_i\}_{i=1}^r$ and $\{y_i\}_{i=1}^r$ may have individual sparse
  supports and need not to be orthogonal. Definition
  \eqref{eq:sandl:Mdef} fits more into the sparse model considered
  also in \cite{Fornasier2018}.  

By $e_i$ we denote $i$-th canonical basis vector and define
row-support $\supp_1(X)$ and column-support $\supp_2(X)$ of a matrix
$X$ as
\begin{equation}
  \supp_1(X)\coloneqq \{i\,:\,\TwoNorm{X^\ad e_i}>0\} \, ,\quad
  \supp_2(X)=\supp_1(X^\ad ) \, .
\end{equation}
Obviously, the matrices $\M^{n_1\times n_2}_{r,s_1,s_2}$ are then
at most $k_1$-column-sparse and $k_2$-row-sparse (sometimes also called as
joint-sparse), i.e.,
\begin{equation}
  |\supp_1(X)|\leq rs_1=:k_1\quad\text{and}\quad |\supp_2(X)|\leq rs_2=:k_2
\end{equation}
\emph{but not} strictly
$k_1k_2=r^2s_1s_2$--sparse\footnote{Assuming that $k_1\coloneqq rs_1\leq n_1$ and
  $k_2\coloneqq rs_2\leq n_2$.}.
Since we have sums of $r$ different
$(s_1,s_2)$-sparse matrices and there are at most $r(s_1\cdot s_2)$
non-zero components. 
Note that a joint
  $(s_1,s_2)$-row and column sparse matrix has only $s_1\cdot s_2$
  non-zero  entries. 
Hence, considering this only from the perspective
of sparse vectors, we expect that up to logarithmic terms recovery can be achieved from
$m\propto r(s_1\cdot s_2)$ measurements. On the other hand, solely from a
viewpoint of low-rankness, $m\propto r(n_1+n_2)$ measurements also
determine an $n_1\times n_2$-matrix of rank $r$. 
Combining both gives therefore  $m\propto r \min(s_1s_2,n_1+n_2)$.

On the other hand, these matrices are determined by at most
$r(s_1+s_2)$ non-zero numbers. Optimistically, we therefore hope that
already $m\lesssim r(s_1+s_2)$ sufficiently diverse observations are
enough to infer on $X$ which is substantially smaller and scales additive in
$s_1$ and $s_2$.  In the next part we will discuss that this intuitive
parameter counting argument is indeed true in the low-rank regime $r\lesssim \log\max(\frac{n_1}{rs_1},\frac{n_2}{rs_2})$.
A generic low-dimensional
embedding of this simultaneously sparse and low-rank structure into
$\KK^m$ for $m\propto r(s_1+s_2)$ via a Gaussian map $\A$ is
stably injective.

\subsection{About RIP for sparse and low-rank matrices}
Intuitively,
$(s_1,s_2)$--sparse rank-$r$ matrices
can be stably identified from $y$ if $\M^{n_1\times n_2}_{r,s_1,s_2}$
is almost-isometrically mapped into $\KK^m$, i.e., distances
between different matrices are preserved up to small error during the
measurements process.  Note that we have the inclusion
\begin{equation}
  \M^{n_1\times n_2}_{r,s_1,s_2} - \M^{n_1\times n_2}_{r,s_1,s_2}
  \subset \M^{n_1\times n_2}_{2r,2s_1,2s_2} \, .
\end{equation}
Since $\A$ is linear it is therefore sufficient to ensure that norms
$\|\A(X)\|\sim \|X\|$ are preserved for $X\in\M^{n_1\times n_2}_{2r,2s_1,2s_2}$.  We
say that a map $\A : \KK^{n_1 \times n_2} \to \KK^m$ satisfies
$(r,s_1,s_2)$-RIP with constant $\delta$ if 
\begin{equation}
  \sup_{\substack{X \in \mathcal{S}}} 
  \TwoNorm{\A^\ad\A(X)-X} \leq \delta%(\A) 
  \label{eq:sandl:rip:sup}
\end{equation}
holds, 
where the supremum is taken over all
$\mathcal{S}=\{X\in\M^{n_1\times n_2}_{r,s_1,s_2}\,:\,\|X\|_2=1\}$ and
$\A^\ad$ denotes the adjoint of $\A$ (defined in the canonical
way with respect to the Hilbert-Schmidt inner product). 
By
$\delta(\A)$ we always denote the smallest value $\delta$ such that above condition holds. 
Equivalently, we have
\begin{equation}
  (1-\delta(\A)) \TwoNorm{X}^2 \leq \lTwoNorm{\A(X)}^2 \leq (1+\delta(\A)) \TwoNorm{X}^2
  \quad \forall X \in \M^{n_1\times n_2}_{r,s_1,s_2} \, .
  \label{eq:sandl:RIP}
\end{equation}
A generic result, based on the ideas of \cite{Baraniuk2008c},
\cite{Candes09:LMR} and \cite{Recht_2010_Guaranteed}, has been
presented already in \cite{Jung2014}. It shows that a random linear
map $\A$ which concentrates uniformly yields the RIP property
\eqref{eq:sandl:RIP} with overwhelming probability (exponential small
outage) once the number of measurements are in the order of the metric
entropy measuring the complexity of the structured set
$\mathcal{S}$. In the case of simultaneous low rank and sparse
matrices this quantity scales (up to logarithmic terms) additively in
the sparsity, as desired.
A version for rank-$r$ matrices where
  $\{x_i\}_{i=1}^r$ and $\{y_i\}_{i=1}^r$ in \eqref{eq:sandl:Mdef} are
  orthonormal sets having joint-sparse supports has been sketched
  already in \cite[Theorem III.7]{Lee2013}, i.e., $\A$ with iid
  Gaussian entries acts almost isometrically in this case for
  $m\gtrsim \delta^{-2}r(s_1+s_2)$ with probability at least $1-\exp(-c_2\delta^2m)$,
$c_2$ being an absolute constant.
Another RIP perspective has been considered in
  \cite{Fornasier2018} where the supremum in \eqref{eq:sandl:rip:sup}
  is taken over effectively sparse ($\{x_i\}_{i=1}^r$ and
  $\{y_i\}_{i=1}^r$ in \eqref{eq:sandl:Mdef} are now only
  well--approximated by sparse vectors) rank-$r$ matrices $X$ with
  $(\sum_{i=1}^r\|x_i\|^2_{\ell_2}\|y_i\|^2_{\ell_2})^{1/2}\leq\Gamma$
  (implying $\TwoNorm{X}\leq\sqrt{r}\cdot\Gamma$). More precisely, for
  $m\gtrsim \Delta^{-2} r(s_1+s_2)$ with $\Delta\in(0,1)$ an operator
  $\A$ with iid.\ centered sub-Gaussian entries acts
  almost-isometrically with probability at least
  $1-2\exp(-C'\Delta m)$ at $\delta=\Delta\Gamma^2 r$ and $C'$ is an
  absolute constant. Note that this probability is slightly weaker.

We provide now a condensed version of the generic statement in \cite{Jung2014}. More precisely,
RIP \eqref{eq:sandl:RIP} is satisfied with high probability for a random linear map $\A$
that has uniformly sub-Gaussian marginals:
\begin{theorem}[RIP for sub-Gaussian measurements]
\label{thm:RIP_low_rank_and_sparse}
  Let $\A:\mathbb{R}^{n_1\times n_2}\rightarrow\mathbb{R}^m$ be a
  random linear map which for given $c>0$ and $0<\delta<1$ fulfills
  $\PP[|\lTwoNorm{\A(X)}-\|X\|_2|\leq\frac{\delta}{2}\|X\|_2]
  \geq
  1-\e^{-c\delta^2m}$ uniformly for all
  $X\in\RR^{n_1\times n_2}$.  If
  \begin{equation}
    m\geq c''\delta^{-2}r(s_1+s_2)\left(1+\log
      \max\{\frac{n_1}{rs_1},\frac{n_2}{rs_2}\}+r\log(9\cdot 4/\delta)\right)
  \end{equation}
  then $\A$ satisfies $(r,s_1,s_2)$-RIP with constant $\delta(\A)\leq \delta$
  with probability $\geq 1-\e^{-\tilde{c}\delta^2 m}$.
  Here, $\tilde c>0$ is an absolute constant and $c''$ is a constant
  depending only on $\delta$ and $\tilde c$. 
\end{theorem}
Clearly, standard Gaussian measurement maps $\A$ fulfill the
concentration assumption in the theorem. More general sub-Gaussian maps
are included as well, see here also the discussion in \cite{Jung2014}.
The proof steps are essentially well-known.  For the sake of
self-containedness we review the steps having our application in
mind.
\begin{proof}
  First, we construct a special $\epsilon$-net
  $\mathcal{R}\subset\mathcal{S}$ for
  $\mathcal{S}=\{X\in\M^{n_1\times
    n_2}_{r,s_1,s_2}\,:\,\|X\|_2=1\}$. By this we mean a set such that
  for each $X\in\mathcal{S}$ we have some $R=R(X)\in\mathcal{R}$ such
  that $\|X-R\|_2\leq\epsilon$.
  Since the matrices $X\in\M^{n_1\times n_2}_{r,s_1,s_2}$ are
  $k_1\coloneqq rs_1$ row-sparse and $k_2\coloneqq rs_2$ column-sparse there are
  \begin{equation}
    L=\binom{n_1}{k_1}\binom{n_2}{k_2}\leq
    \left(\frac{\e n_1}{k_1}\right)^{k_1}
    \left(\frac{\e n_2}{k_2}\right)^{k_2}
    \leq\left(\e\max\{\frac{n_1}{k_1},\frac{n_2}{k_2}\}\right)^{k_1+k_2}
  \end{equation}
  different combinations for the row support $T_1\subset[n_1]$ and
  column support $T_2\subset[n_2]$ with $|T_1|=k_1$ and $|T_2|=k_2$.
  For each of these canonical matrix subspaces supported on
  $T_1\times T_2$, we consider matrices of rank at most $r$.  From
  \cite[Lemma 3]{Candes09:LMR} it is known that there exists an
  $\epsilon$-net for $k_1\times k_2$ matrices of rank $r$ of
  cardinality $(9/\epsilon)^{(k_1+k_2)r}$ giving therefore
  \begin{equation}
    \log|\mathcal{R}|\leq (k_1+k_2)\left(1+\log \max\{\frac{n_1}{k_1},\frac{n_2}{k_2}\}+r\log(9/\epsilon)\right) \, .
    \label{eq:sandl:Heps}
  \end{equation}
  In other words, up to logarithmic factors, this quantity also
  reflects the intuitive parameter counting.  The net construction
  also ensures that for each $X\in \mc S$ and close by net point $R=R(X)$ we have $|\supp_1(X-R)|\leq k_1$ and
  $|\supp_2(X-R)|\leq k_2$.  However, note that in non-trivial cases $\rank(R-X)=2r$
  meaning that
  $(R-X)/\|R-X\|_2\notin\mathcal{S}$.  But, using a singular value
  decomposition one can find $R-X=X_1+X_2$ with
  $\langle X_1,X_2\rangle=0$ for some $X_1/\|X_1\|_2\in\mathcal{S}$ and
  $X_2/\|X_2\|_2\in\mathcal{S}$.  To show RIP, we define the constant
  \begin{equation}
    A\coloneqq \max_{X\in\mathcal{S}}|\lTwoNorm{\A(X)}-1| \, .
    \label{eq:sandl:defA}
  \end{equation}
  For some $X\in\mathcal{S}$ and close by net point
  $R=R(X)\in\mathcal{R}$
  and let us consider $\delta$ such that 
  $\lTwoNorm{\A(R)}-1|\leq\delta/2$. 
  Then, 
  \begin{equation}
    \begin{split}
      |\|&\A(X)\|_{\ell_2}-1|
      \leq |\lTwoNorm{\A(X)}-\lTwoNorm{\A(R)}|+|\lTwoNorm{\A(R)}-1|\\
      &\leq\|\A(X-R)\|_{\ell_2}+\frac{\delta}{2}
      \leq\|\A(X_1)\|_{\ell_2}+\|\A(X_2)\|_{\ell_2}+\frac{\delta}{2}\\
      &\leq (1+A)(\|X_1\|_2+\|X_2\|_2)+\frac{\delta}{2}
      = (1+A)\|X_1+X_2\|_2+\frac{\delta}{2}\\
      &=(1+A)\|X-R\|_2+\frac{\delta}{2}
      \leq(1+A)\epsilon+\frac{\delta}{2}\, .
    \end{split}
  \end{equation}
  Now we choose $\tilde{X}\in\mathcal{S}$ satisfying
  $A=|\|\A(\tilde{X})\|_{\ell_2}-1|$ ($\mathcal{S}$ in
  \eqref{eq:sandl:defA} is compact).
  For such an $\tilde{X}$ we also have
  \begin{equation}
    \begin{split}
      A=|\|&\A(\tilde{X})\|_{\ell_2}-1|  
      \leq(1+A)\epsilon+\frac{\delta}{2} \, .
    \end{split}
  \end{equation}
  Requiring that the right hand side is bounded by $\delta$ and solving this inequality for $A$ (assuming $\epsilon<1$) we find that indeed
  $A\leq\frac{\epsilon+\delta/2}{1-\epsilon}
  \leq \delta$ whenever $\epsilon\leq\delta/(2+2\delta)$. 
  In particular we can choose $\delta<1$ and we set $\epsilon=\delta/4$. 
  Therefore, \eqref{eq:sandl:Heps} yields
  \begin{equation}
    \begin{split}
      \log |\mathcal{R}|\leq (k_1+k_2)\left(1+\log
        \max\{\frac{n_1}{k_1},\frac{n_2}{k_2}\}+r\log(9\cdot 4/\delta)\right) .
    \end{split}
  \end{equation}
  Using the assumption
  $\PP[|\lTwoNorm{\A(X)}-\TwoNorm{X}|\leq\delta/2\TwoNorm{X}]\geq
  1-\e^{-c\delta^2m}$ 
  and the union bound
  $\PP[\forall R\in\R\,:\,|\lTwoNorm{\A(R)}-1|\leq\frac{\delta}{2}]
  \leq
    1-\e^{-(c\delta^2m-\log|\mathcal{R}|)}$
  we obtain RIP with probability at least
  \begin{equation}
    \PP[\forall X\in\M^{n_1\times n_2}_{r,s_1,s_2}\,:\,|\lTwoNorm{\A(X)}^2-\|X\|^2_2|\leq
    \delta\|X\|^2_2]\geq
    1-\e^{-(c\delta^2m-\log|\mathcal{R}|)} \, .
  \end{equation}
  Thus, if we want to have RIP satisfied with probability $\geq 1-\e^{\tilde{c}\delta^2 m}$
  for a given $\tilde{c}>0$, i.e., 
  \begin{equation}
    c\delta^2m-\log|\mathcal{R}|=\delta^2m(c-\frac{\delta^{-2}\log|\mathcal{R}|}{m})\overset{!}{\geq}
    \tilde{c}\delta^2 m \, ,
  \end{equation}
  it is sufficient to impose that $m\geq c''\delta^{-2}\log|R|$ for a some $c''>0$. 
\qed\end{proof}
In essence the theorem shows that the intrinsic geometry of sparse and
low-rank matrices is preserved in low-dimensional embeddings when
choosing the dimension above a threshold.  
It states that, in the low-rank regime
$r\lesssim \log\max(\frac{n_1}{rs_1},\frac{n_2}{rs_2})$, for fixed
$\delta$ this threshold the RIP to hold scales indeed as $r(s_1+s_2)$.
  % We believe that 
  This additional low-rank
  restriction is an technical artifact due to suboptimal combining of
  covering number estimates.  Indeed, upon revising the manuscript we
  found that the statement above may be improved by utilizing
  \cite[Lemma 4.2]{Fornasier2018} instead of \eqref{eq:sandl:Heps} which yields a scaling of $r(s_1+s_2)$ without restrictions on $r$ and with
  probability of at least $\geq 1-\exp(-\tilde{c}\delta^2m)$.  From
  the proof it follows also easily that for joint-sparse matrices
  where each of the sets $\{x_i\}_{i=1}^r$ and $\{y_i\}_{i=1}^r$ in
  \eqref{eq:sandl:Mdef} have also joint support as in \cite{Lee2013},
  a sampling rate $m \propto r(s_1+s_2)$ is sufficient anyway for all
  ranks $r$. 
  % $r\leq\min(n_1,n_2)$.  
%

Intuitively, one should
therefore be able reconstruct an unknown $s_1\times s_2$--sparse
matrix of rank $r$ from $m\propto r(s_1+s_2)$ generic random measurements.
This would indeed reflect the intuitive parameter counting argument.
Unfortunately, as will be discussed next, so far, no algorithm is
known that can achieve such a reconstruction for generic matrices.

\subsection{Some more details on related work}
It is well-known that sufficiently small RIP constants $\delta(\A)$
imply successful convex recovery for sparse vectors \cite{Candes2005b} and
low-rank matrices \cite{Gro11,Recht_2010_Guaranteed}, separately.
An intuitive starting point for convex recovery of the elements from
$\M^{n_1\times n_2}_{r,s_1,s_2}$ would therefore be the program:
\begin{equation}
  \min \mu_1\|X\|_{1}+\mu_{\ell_1}\|X\|_{\ell_1}\quad\text{subject to } \A(X) = y
  \label{eq:sandl:sum}
\end{equation}
which uses a weighted sum as a regularizer, where $y=\A(X_0)$ are
noiseless measurements of the signal $X_0$. Related approaches
have been used also for applications including sparse phase retrieval
and sparse blind deconvolution. Obviously, then the corresponding
measurement map is different and depends on the particular
application.  The practical relevance of this convex formulation is
that it always allows to use generic solvers and there is a rich
theory available to analyze the performance for certain types
measurement maps in terms of norms of the recovery error
$X-X_0$. Intuitively, one might think that this amounts only to
characterize the probability when the matrix $A$ is robustly injective
on feasible differences $X-X_0$, i.e., fulfills RIP or similar
conditions. 
However, this is not enough as observed and worked out in
\cite{Oymak2015,Jalali2016}.  One of the famous no-go results in these
works is that no extra reduction in the scaling of the sampling rate
can be expected as compared to the best of recovering with respect to
either the low-rank structure ($\mu_{\ell_1}=0$) or sparsity
($\mu_{1}=0$), separately.  
In other words, for any pair $(\mu_{\ell_1}, \mu_{1})$ the required sampling rate can not be better than the minimum of the one for $\mu_{\ell_1}=0$ and $\mu_{\ell_1}=0$. 
A difficult point in this discussion is what will happen if the program
is optimally tuned, i.e., if $\mu_{\ell_1}=1/\|X_0\|_{\ell_1}$ and
$\mu_{1}=1/\|X_0\|_{1}$. We have based our generic investigations
given in Section~\ref{subsec:generic:lowerbounds} on the considerably
more simplified technique of \cite{MuHuaWri13} which also allows to
obtain such results in more generality. An alternative convex approach
is discussed \cite{RicOboVer14} where the corresponding atomic norm
\cite{Chandrasekaran2010} (called $kq$-norm) is used as a single
regularizer.  This leads to convex recovery at optimal sampling rate
but the norm itself cannot be computed in a tractable manner,
reflecting again the hardness of the problem itself.  For certain
restricted classes the hardness is not present and convex algorithms
perform optimally, see exemplary \cite{Aghasi2017} where signs in a
particular basis are known a-priori.

Due to the inability of tractable convex programs non-convex recovery
approaches have been investigated intensively in the last years. In
particular, the alternating and thresholding based algorithm ``sparse
power factorization'', as presented in \cite{Lee2013, Lee2015}, can provably
recover at optimal sampling rates when initialized optimally.
However, this is again indeed the magic and difficult step since
computing the optimal initialization is again computationally
intractable.
 For a suboptimal but tractable
initialization recovery can only be guaranteed for a considerable
restricted set of very peaky signals. Relaxed conditions have been
worked out recently \cite{Geppert2018} with the added benefit that
the intrinsic balance between additivity and multiplicativity in
sparsity is more explicitly established. 
Further alternating
algorithms like \cite{Fornasier2018} have been proposed with
guaranteed local convergence and which have better empirical
performance.

An interesting point has been discussed in
\cite{Foucart2019}. Let for simplicity $n=n_1=n_2$ and $s=s_1=s_2$.
Assume that for given rank $r$ and sparsity $s$ the measurement map
in \eqref{eq:sandl:meas} factorizes in the form
$A_i=B^\ad\tilde{A}_i B\in\mathbb{R}^{n\times n}$ where
$\tilde{A}_i\in\mathbb{R}^{p\times p}$ for $i=1,\dots,m\simeq rp$
and $B\in\mathbb{R}^{p\times n}$ are all independent standard
Gaussian matrices with $p\simeq s\log(en/s)$. In this case a
possible reconstruction approach will factorize as well into two
steps, (i) recovery of an intermediate matrix
$Y\in\mathbb{R}^{p\times p}$ from the raw measurements $y$ using
nuclear norm minimization and (ii) recovery of the unknown matrix
$X_0$ from $Y$ using the HiHTP algorithm (details see
\cite{Roth2016}).  However, in the general case, hard-thresholding
algorithms like HiHTP require computable and almost-exact (constants
almost independent of $s$ and $n$) head projections into
$(s,s)$-sparse matrices. Positive-semidefiniteness is helpful in
this respect \cite{Foucart2019} and in
particular in the rank-one case this is relevant for sparse phase
retrieval \cite{Iwen:2017}. 
But in the general case, to the best of the authors knowledge, no algorithm
with tractable initialization has guaranteed global convergence for
generic sparse and low-rank matrices so far.

\subsection{The lower bound}
In the following section we will further strengthen this ``no-go''
result for convex recovery. As already mentioned above, an issue which
has not been discussed in sufficient depth is what can be said about
optimally tuned convex programs and beyond convex combinations of
multiple regularizers. 
For our simultaneously sparse and low rank matrices Theorem \ref{thm:generic:lowerbound:max} yields the following. 

\begin{corollary}[Lower bound, sparse and low rank matrices]
  Let $0\neq X_0\in\M^{n_1,n_2}_{r,s_1,s_2}$ be an $(s_1,s_2)$-sparse
  $n_1\times n_2$-matrix of rank at most $r$,
  $\bar{n}\coloneqq \frac{n_1n_2}{\min(n_1,n_2)}$ and
  $\A:\RR^{n_1\times n_2}\rightarrow \RR^m$ be a Gaussian measurement operator.
  Then, for all $m\leq r\min(\bar{n},s_1s_2)-2$, $X_0$ is the unique
  minimizer of
  \begin{equation}
    \min \max\Bigl\{\frac{\|X\|_{\ell_1}}{\|X_0\|_{\ell_1}},
    \frac{\|X\|_{1}}{\|X_0\|_{1}}\Bigr\}\quad\text{subject to}\quad \A(X) = \A(X_0)
    \label{eq:cor:sl:lowerbound:max:cvx}
  \end{equation}  
  with probability at most
  \begin{equation}
    \psuccmax 
    \leq 
    4 \exp\left(-\frac{(r\min(\bar{n},s_1s_2)-m-2)^2}{8\kappa} \right) \, .
  \end{equation}
  \label{cor:sl:lowerbound:max}
\end{corollary}
In words, even when optimally tuning convex algorithms and when using the
intuitive best regularizer having the largest subdifferential, the
required sampling rate still scales multiplicative in sparsity, i.e.,
it shows the same no-go behavior as the other (suboptimal) regularizer.
\begin{proof}
The Lipschitz constants of the $\ell_1$--norm and the nuclear norm (w.r.t.\ the Frobenius norm) are 
\begin{align}
L_{\ell_1} = \sqrt{n_1n_2} 
\qquad \text{and}\quad
L_1 = \sqrt{\min(n_1,n_2)} \, ,
\end{align}
respectively. 
Using that the matrix $X_0$ is $s\coloneqq r(s_1s_2)$--sparse yields
$ \|X_0\|_{\ell_1}\leq\sqrt{s}\cdot\|X_0\|_{\ell_2}$. 
Hence, 
\begin{equation}
  \kappa_{\ell_1}\geq\frac{n_1n_2\cdot s}{L^2_{\ell_1}}-2=s-2
\end{equation}
is the expression in the minimum in
\eqref{eq:min_nr_of_measurements} corresponding to the index
``$(i)=\ell_1$'' used for the $\ell_1$ norm.  Using that $X_0$ has rank
at most $r$ we obtain $\|X_0\|_{1}\leq \sqrt{r}\|X_0\|_{\ell_2}$.
Hence,
\begin{equation}
  \kappa_{1}\geq\frac{n_1n_2\cdot r}{L^2_{1}}-2=\frac{n_1n_2\cdot r}{\min(n_1,n_2)}-2
  \eqqcolon \bar{n}\cdot r-2
\end{equation}
with $\bar{n}\coloneqq \frac{n_1n_2}{\min(n_1,n_2)}$ is the expression
in the minimum in \eqref{eq:min_nr_of_measurements} corresponding to
the index ``$(i)=1$'' used for the nuclear norm.
Together, 
\begin{equation}
  \kappa= \min(\kappa_{\ell_1},\kappa_1)\geq \min(\bar{n}r,s)-2
  =r\min(\bar{n},s_1s_2)-2 
\end{equation}
and Theorem \ref{thm:generic:lowerbound:max} establishes the corollary. 
\qed\end{proof}

\begin{figure}
  \centering
  
  % \includegraphics[width=.48\linewidth]{Fig1a_ConicDistO1_s=2}
  % \includegraphics[width=.48\linewidth]{Fig1b_ConicDistO1_s=3}
  % \includegraphics[width=.48\linewidth]{Fig1c_ConicDistO1_s=4}
  % \includegraphics[width=.48\linewidth]{Fig1d_ConicDistO1_s=5}
  
  % \caption{Statistical dimension from \eqref{eq:stad_dim_as_dist} for $n\times n$ rank-one and
  %   $s\times s$--sparse matrices $\M_{1,s,s}^{n\times n}$ for
  %   $s=2,3,4,5$ and $n=10,\dots, 50$. The results are
  %   obtained by averaging the solutions of the corresponding SDP's like
  %   e.g. \eqref{eq:sandl:Gaussian_dist_sum_SDP} for the sum.
  %   \emph{comment: the numerical experiments are still running and the results will
  %     be replaced in final version.}
  % }
  \includegraphics[width=.9\linewidth]{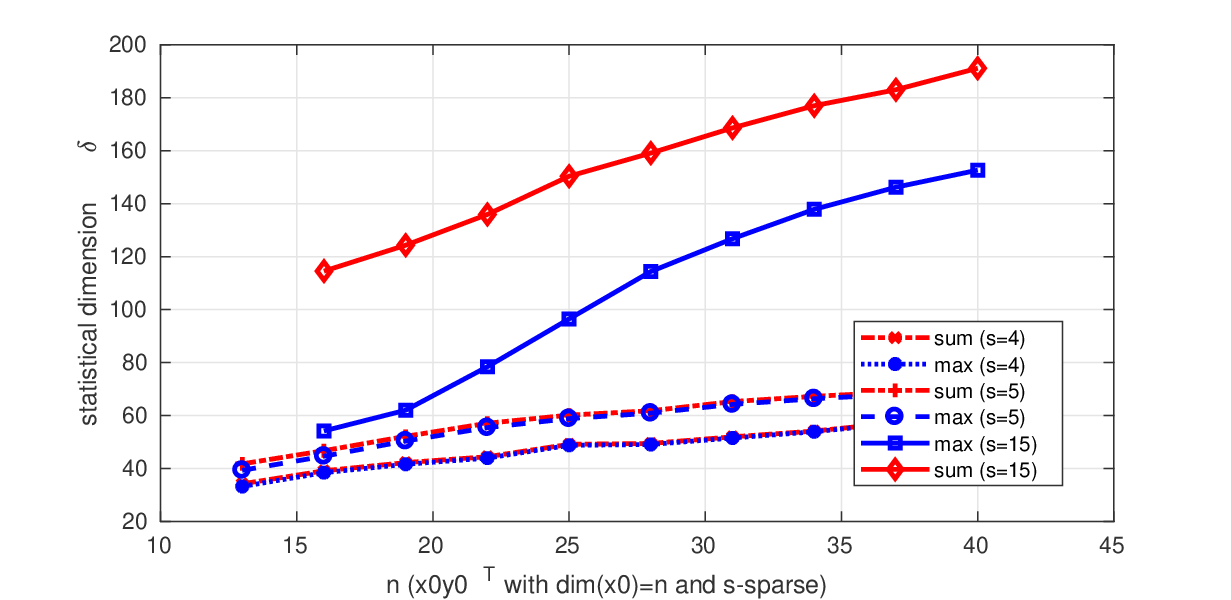}

  \caption{Statistical dimension from \eqref{eq:stad_dim_as_dist} for $n\times n$ rank-one and
    $s\times s$--sparse matrices $\M_{1,s,s}^{n\times n}$ for
    $s=4,5,15$ and $n\in [15,40]$. The results are
    obtained by averaging the solutions of the corresponding SDP's like
    e.g. \eqref{eq:sandl:Gaussian_dist_sum_SDP} for the sum.
  }
  \label{fig:sandl:cdists}
\end{figure}

\subsection{Numerical experiments}
We have numerically estimated the statistical dimension of the decent cones and have performed the actual reconstruction of simultaneously low rank and sparse matrices. 

\subsubsection{Gaussian distance}
In Section \ref{subsec:generic:gaussiandist:SDP} we showed that the
Gaussian distance can be estimated numerically by sampling over (in this case) semidefinite programs (SDP) according to
\eqref{eq:generic:Gaussian_dist_sum_SDP} and
\eqref{eq:generic:Gaussian_dist_max_SDP}. When empirically averaging
these results according to \eqref{eq:stad_dim_as_dist} one obtains an estimate of the statistical dimension and therefore the phase
transition point for successful convex recovery. 
% We define the weights as $\mu=(\mu_1,\mu_2)$. 
For the case of
sparse and low-rank matrices with the weighted sum of nuclear norm and $\ell_1$--norm as regularizer the distance \eqref{eq:generic:Gaussian_dist_sum_SDP} becomes
\begin{equation}
  \begin{split}
    &\fNorm{G - \cone \partial \regNormSum{\argdot}(X_0)}\\
    &=
    \min_{\tau\geq 0,X_1,X_2}
    \left\{
      \fNorm{G - X_1-X_2}: \ 
      \langle X_1+X_2, X_0 \rangle =
      \tau\langle
      \lambda,
      \left(\begin{matrix} \norm{X_0}_1\\\norm{X_0}_{\ell_1}\end{matrix}\right)
      \rangle,
      \left(\begin{matrix} \norm{X_1}_\infty\\\norm{X_2}_{\ell_\infty}\end{matrix}\right)
      \leq \tau \lambda
    \right\} \, .
    \label{eq:sandl:Gaussian_dist_sum_SDP}
  \end{split}
\end{equation}
%\textcolor{red}{refine \eqref{eq:sandl:Gaussian_dist_sum_SDP} for optimal weighting}

A similar SDP can be obtained for the case of the maximum of these two
regularizers.  We solve both SDPs using the CVX toolbox in MATLAB
(with SDPT3 as solver) for many realization of a Gaussian matrix $G$
and then average those results.  We show such results for the optimal
weights in Figure \ref{fig:sandl:cdists} for $\M_{1,s,s}^{n\times n}$
where $s=4,5,15$ and the size of the $n\times n$ matrices ranges in
$n\in[15,40]$. For $s=4,5$ the statistical dimension for the optimally
weighted sum and the maximum are almost the same. However, for higher
sparsity $s=15$ there is a substantial difference, i.e., the optimally
weighted sum of regularizers behaves worse than the maximum.

% $s=2,3,4,5$ and the size of the $n\times n$ matrices ranges in
% $n=10,\dots,50$.

\begin{figure}
  \centering
    % \includegraphics[width=.48\linewidth]{phasetrans_n=30_s=5-20-l1_20190211_191146_l1norm-ssucc}
    % %
    % \includegraphics[width=.48\linewidth]{phasetrans_n=30_s=5-20-nuc_20190211_191110_nucnorm-ssucc}
    % %
    % \includegraphics[width=.48\linewidth]{Fig2b_phasetrans_n=30_s=5-20_20180102_122313_sum_opt_weights-ssucc}
    % %
    % \includegraphics[width=.48\linewidth]{Fig2a_phasetrans_n=30_s=5-20_20180102_122313_max_opt_weights-ssucc}
    % %
    % \includegraphics[width=.48\linewidth]{phasetrans_n=30_s=5-20-prox_20180102_122355_sum_prox_weights-ssucc}
    % %
    % \includegraphics[width=.48\linewidth]{phasetrans_n=30_s=5-20-spf_20190117_143341_spf-ssucc}
    % recreate by
    % (1) latex compiling pics_prep/fig_sandl/phasetrans/phasetrans.tex
    % (2) running dvips -E on the dvi-file 
    % (3) eps2png
  %
  \iffrontiers\else
  \begin{adjustwidth}{-1cm}{-1cm}
  \fi
  \includegraphics[width=\linewidth]{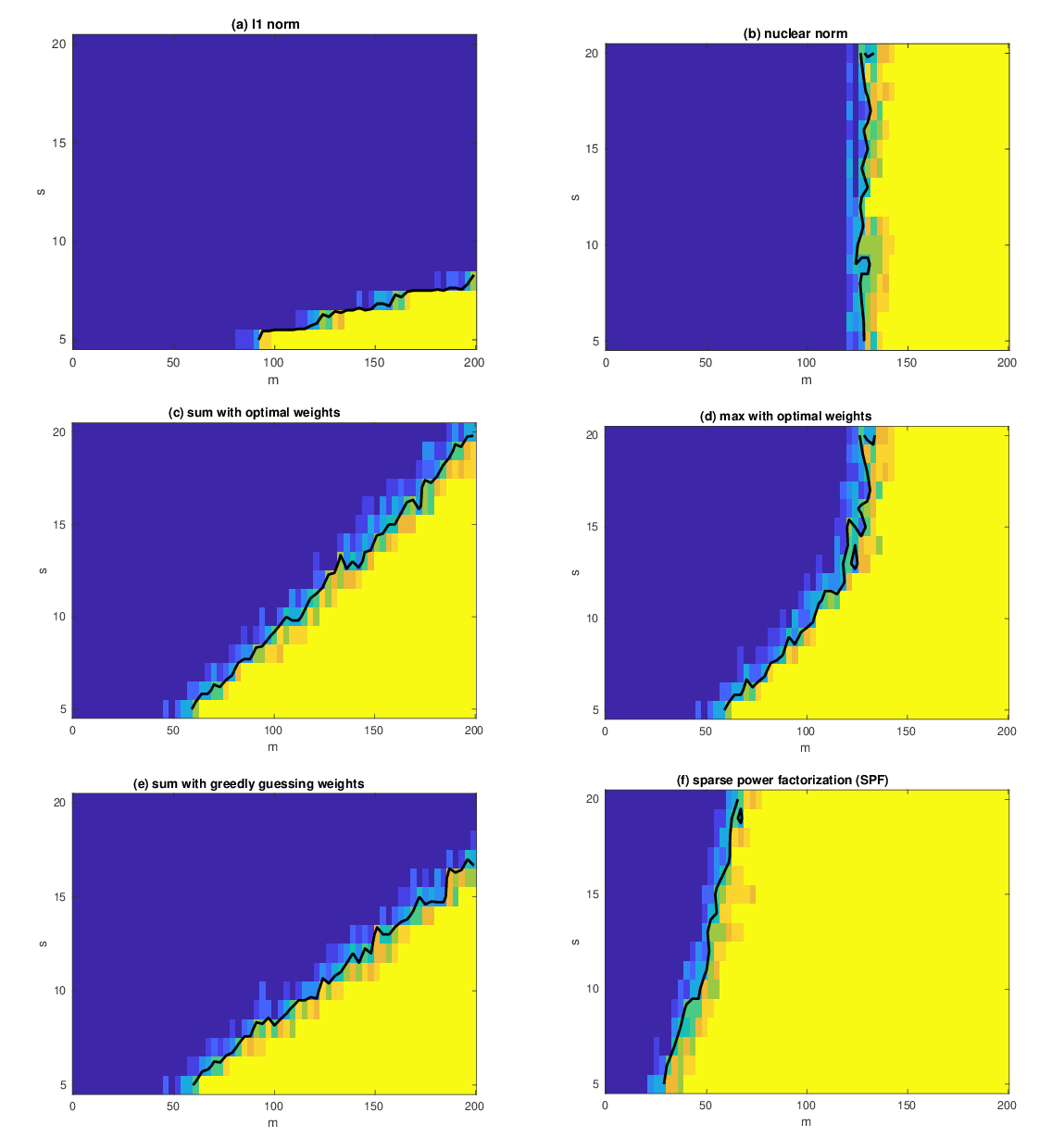}
  \iffrontiers\else
  \end{adjustwidth}
  \fi

    \caption{Phase transitions for convex recovery using the
      $\ell_1$--norm (a), nuclear norm (b), the
      max-norm (d) and the sum-norm (c) with optimal
      $X_0$-dependent weights as regularizers. Furthermore, guessing
      weights in a greedy fashion for the sum-norm using %Algorithm
      % \ref{algo:SumProxWeights}
      is shown (e) and non-convex recovery using sparse power factorization (SPF) from \cite{Lee2013}
      is in (f).
    }
  \label{fig:sandl:phasetrans}
\end{figure}

These results indeed show that the statistical dimension for optimally
weighted maximum of regularizers is better than the sum of
regularizers. 

\subsubsection{Convex recovery}
We numerically find the phase transition for the
convex recovery of complex sparse and low-rank matrices using the sum
and maximum of optimally weighted $\ell_1$ and nuclear norm. 
We also compare to the results obtained by convex
recovery using only either the $\ell_1$--norm or the nuclear norm as reguralizer and,
exemplary, also to a non-convex algorithm. 

The dimension of the matrices are $n=n_1=n_2=30$, the
sparsity range is $s=s_1=s_2=5,\dots,20$ and the rank is $r=1$.  For
each parameter setup a matrix $X_0=uv^\ad$ is drawn using a
uniform distribution for supports of $u$ and $v$ of size $s$ and
iid.\ standard complex-normal distributed entries on those support. The
measurement map itself also consists also of iid.\ complex-normal
distributed entries.  The reconstructed vector $X$ is obtained using
again the CVX toolbox in MATLAB with the SDPT3 solver and an
reconstruction is marked to be successful exactly if
\begin{equation}
  \|X-X_0\|_2/\|X\|_2\leq 10^{-5} 
\end{equation}
holds. Each $(m,s)$-bin in the phase transition plots contains $20$ runs.

The results are shown in Figure~\ref{fig:sandl:phasetrans}. 
Plots~(a) and~(b) show the phase transition of only taking the
$\ell_1$--norm and the nuclear norm as reguralizer, respectively. 
The lower bound
from Theorem~\ref{thm:generic:lowerbound:max} on the required number
of measurements yield for those cases the sparsity
$s^2$ of $X_0$ and $n$, respectively. Thus, only for very small
values of $s^2$ there is a clear advantage of
$\ell_1$-regularization compared to the nuclear norm.
The actual recovery rates scale, however, are close to $2 s^2 \ln(n^2/s^2)$ and $4rn$.

However, combining both regularizers with optimal weights improves as
shown in Plots~(c) and~(d) of Figure~\ref{fig:sandl:phasetrans}. 
Both combined approaches instantaneously balance between sparsity and rank. 
%
% \mk{In particular, taking the convex combinations
%   of norms, can give a sampling rate that is relatively close to the lower bound of essentially $\min\{s^2,n\}$.}
%
Moreover, there is a clear
advantage of taking the maximum (Plot~(c)) over of the sum (Plot~(d))
of $\ell_1$-- and nuclear norm.  For sufficiently small sparsity the
$\ell_1$--norm is more dominant and for higher sparsity than the
nuclear norm determines the behavior of the phase transition. But only
for the maximum of the regularizers there is the the sampling rate saturates at 
approximately $m=130$ due to $\rank(X_0) = 1$, see Plot~(d).

We also mention that the maximum of regularizers improves only if it is optimally
tuned, which is already indicated by the subdifferential of a maximum \eqref{eq:subdiff:max}, where only the largest terms contribute. 
In contrast, reconstruction behaviour of the sum of norms seems to more stable. 
This observations has also been
mentioned in \cite{Jalali2016}. This feature motivates an empirical
approach of guessing the weights from observations.

To sketch an greedy approach for guessing the weights we consider the
following strategy. %we consider Algorithm~\ref{algo:SumProxWeights}.
Ideally, we would like to choose $\lambda_1 = 1/\OneNorm{X_0}$ and
$\lambda_{\ell_1} = 1/\lOneNorm{X_0}$ in the minimization of the objective function 
$\lambda_1 \OneNorm{X} + \lambda_{\ell_1} \lOneNorm{X}$.  
Since, for Frobenius norm
normalized $X$ we have $1/\OneNorm{X} \geq \InfNorm{X}$ (similarly for
the $\ell_\infty$-norm) we choose for as initialization
$\lambda_1^{(1)}=\lInfNorm{\A^\ad(y)}$ and
$\lambda_{\ell_1}^{(1)}=\InfNorm{A^\ad(y)}$ for the iteration $t=1$. 
After finding
\begin{equation}
  X^{(t)}\coloneqq \argmin_X\{ \lambda_1^{(t)} \OneNorm{X} + \lambda_{\ell_1}^{(t)} \lOneNorm{X}:\ \A(X) = y \}
\end{equation}
we update $\lambda_1^{(t+1)} \coloneqq 1/\OneNorm{X^{(t)}}$ and
$\lambda_{\ell_1}^{(t+1)} \coloneqq 1/\lOneNorm{X^{(t)}}$.  The results obtained by this greedy
approach after $3$ iterations are shown in Plot~(e) of Figure~\ref{fig:sandl:phasetrans}. 
Comparing this to the optimally
weighted sum of regularizers in Plot~(c), we see that almost the same
performance can be achieved with this iterative scheme. 

% So, in Algorithm~\ref{algo:SumProxWeights} we use lower bounds on $1/\OneNorm{X_0}$ and $1/\lOneNorm{X_0}$ in the minimization, as a relaxation. 
% In contrast, $1/\OneNorm{X}$ and $1/\lOneNorm{X}$ are upper bounds on $1/\OneNorm{X_0}$ and $1/\lOneNorm{X_0}$\mk{ (right?)}.  
% {\small
% \begin{algorithm}[h]
%   \begin{algorithmic}
%     \State \textbf{Input:} $\A\in \CC^{m \times n}$, $y\in \CC^m$
%     \State \textbf{Output:} $X \in \CC^{n\times n}$ (low rank and sparse) such that $\A x = y$
%     \State
%     \Function{SumProxWeights$(\A,y)$}{}
%     \State $X^{(0)} \gets \A^\ast y$
%     \State $j = 0$ 
%     \While{Stopping criterion false}
%     \State $a = \InfNorm{X^{(j)}}$
%     \State $b = \lInfNorm{X^{(j)}}$
%     \State $j = j+1$
%     \State $X^{(j)} \gets CS(\A,y,a,b)$
%     \EndWhile
%     \State \Return $X^{(j)}$
%     \EndFunction
%     \State
%     \Function{$CS(\A,y,a,b)$}{}
%     \State $X \gets \argmin_X\{ a \OneNorm{X} + b \lOneNorm{X}:\ \A(X) = y \}$
%     \State \Return $X$
%     \EndFunction
%   \end{algorithmic}
%   \caption{SumProxWeights}
%   \label{algo:SumProxWeights}
% \end{algorithm}
% }

Finally, there is indeed strong evidence that in many problems with
simultaneous structures non-convex algorithms perform considerably
better and faster then convex formulations.  Although we have focused
in this work on better understand of convex recovery we would bring
also here an example.  For the sparse and low-rank setting there
exists several very efficient and powerful algorithms,
exemplary we mention here sparse power factorization (SPF)
\cite{Lee2013} and ATLAS$_{2,1}$ \cite{Fornasier2018}. In the noiseless
setting SPF is known to clearly outperforms all convex algorithms, see
Plot~(f) of Figure~\ref{fig:sandl:phasetrans}. The numerical
experiments in \cite{Fornasier2018} suggests that in the noisy setting
ATLAS$_{2,1}$ seems to be better choice.

% \begin{figure}
%   \centering
%     %
%     \includegraphics[width=.48\linewidth]{Fig2b_phasetrans_n=30_s=5-20_20180102_122313_sum_opt_weights-ssucc}
%     %
%     \includegraphics[width=.48\linewidth]{phasetrans_n=30_s=5-20-prox_20180102_122355_sum_prox_weights-ssucc}

%     \caption{Phase transition for convex recovery using the sum-norm with
%       (left) optimal $X_0$-dependent weights as regularizers and
%       (right) the results obtained by guessing the weights using
%       Algorithm \ref{algo:SumProxWeights}.
%       % Successful recovery of matrices from
%       % $\M_{1,s,s}^{n\times n}$ for $n=30$ starts again for $s=5$ with
%       % approximately $m=50$ measurements. At approximately $m=130$ the
%       % bound related to the nuclear norm is active for the maximum
%       % (left) and therefore improves over the sum (right).
%     }
%   \label{fig:sandl:phasetrans:prox}
% \end{figure}

%%% ----------------------------------------------------

%%% ===============================================
\section{Special low-rank tensors} % ==============
%%% ===============================================
\label{sec:tensors}
Tensor recovery is an important and notoriously difficult problem, which can be seen as a generalization of low-rank matrix recovery. 
However, for tensors there are several notions of rank and corresponding tensor decompositions \cite{KolBad09}.  
They include 
the higher order singular value decomposition (HOSVD), 
the tensor train (TT) decomposition (a.k.a.\ by matrix product states), 
the hierarchical Tucker (a.k.a.\ tree tensor network) decomposition, 
and the CP decomposition. 
For all these notions, the unit rank objects coincide and are given by tensor products. 

Gandy, Recht, and Yamada \cite{GanRecYam11} suggested to use a sum of nuclear norms of different matrizations (see below) as a regularizer for the completion of $3$-way tensors in image recovery problems. 
Mu et al.~\cite{MuHuaWri13} showed that this approach leads to the same scaling in the number of required measurements as when one just one nuclear norm of one matrization as a regularizer. 
However, the prefactors are significantly different in these approaches. 
Moreover, Mu et al.\ suggested to analyze $4$-way tensors, where the matrization can be chosen such that the matrices are close to being square matrices. 
In this case, the nuclear norm regularization yields an efficient reconstruction method with rigorous guarantees that has the so far best scaling in the number of measurements. 
For rank-$1$ tensors we will now suggest to use a maximum of nuclear norms of certain matrizations as a regularizers. 
While the no-go results \cite{MuHuaWri13} for an optimal scaling still hold, this still leads to a significant improvement of prefactors. 

%%% ------------------------------------------------
\subsection{Setting and preliminaries} % -----------
%%% ------------------------------------------------
The \emph{effective rank} of a matrix $X$ is 
$\rank^\eff(X) \coloneqq \norm{X}_1^2/\norm{X}_2^2$. 
Note that for matrices where all non-zero singular values coincide, the rank coincides with the effective rank and for all other matrices the effective rank is smaller. 

We consider the tensor spaces $V \coloneqq \RR^{n_1\times n_2 \times \dots \times n_L}$ as signal space and refer to the $n_i$ as \emph{local dimensions}. 
The different matrix ranks of different matrizations are given as follows. 

An index \emph{bipartition} is 
\begin{equation}
\label{bps1}
 [L] = b \cup b^c \quad \textnormal{with} \quad b \subset [L] \quad \textnormal{and} \quad b^c = [L]\backslash b.
\end{equation}
The \emph{$b$-matricization} is the canonical isomorphism $\KK^{n_1 \times n_2 \times \ldots \times n_L} \cong \KK^{n_b \times n_{b^c}}$,
where $n_b = \prod_{i \in b} n_i$, i.e. the indices in $b$ are joined together into the row index of a matrix and the indices in $b^c$ into the column index.
It is performed by a \texttt{reshape} function in many numerics packages. 
The rank and effective of the $b$-matrization of $X$ are denoted by $\rank_b(X)$ and $\rank^\eff_b(X)$. 
The \emph{$b$-nuclear norm} $\norm{X}_1^b$ is given by the nuclear norm of the $b$-matricization of $X$. 

Now, we consider ranks based on a set of index bipartitions
\begin{equation}
  \mc B =( b_j)_{j\in [k]} \, \quad \text{with } \ b_j \subset [L] \, .
\end{equation} 
The corresponding (formal) rank $\rank_{\mc B}$ is given by
\begin{equation}
  \rank_{\mc B}(X) \coloneqq ( \rank_b(X) )_{b \in \mc B}
\end{equation}
Similarly, given a signal $X_0 \in V$ the corresponding max-norm is given by
\begin{equation}\label{eq:max_norm_tensors}
  \regNormMax[\mu^\ast\!]{X}^{\mc B} 
  \coloneqq 
  \max_{b\in \mc B}\,  \frac{\OneNorm{X}^{b}}{\OneNorm{X_0}^{b}} 
  \, .
\end{equation}
Note that for the case that $X_0$ is a product $X_0 = x_0^{(1)}\otimes \dots x_0^{(L)}$ we have 
\begin{equation}
  \norm{X_0}_p^b = \prod_{j=1}^L \lTwoNorm{x_0^{(j)}}
\end{equation}
for all $b \subset [L]$ and $p\geq 1$. 
Hence, the reweighting in the optimal max-norm \eqref{eq:max_norm_tensors} is trivial in that case, i.e., it just yields an overall factor, which can be pulled out of the maximum. 

Let us give more explicit examples for the set of bipartitions:
$\mc B = (\{i\})_{i \in [L]}$ defines the HOSVD rank and 
$\mc B = (\{1, \dots, \ell \})_{\ell \in [L-1]}$ the tensor train rank. 
They also come along with a corresponding tensor decomposition. 
In other cases, accompanying tensor decompositions are not known. 
For instance, for $k=4$ and $\mc B = (\{1,2\}, \{1,3\})$ it is clear that the tensors of (formal) rank $(1,1)$ are tensor products. 
The tensors of ranks $(1,i)$ and $(i,1)$ are given by tensor products of two matrices, each of rank bounded by $i$. 
But in general, it is unclear what tensor decomposition corresponds to a $\mc B$-rank. 

One interesting remark might be that there are several measures of entanglement in quantum physics, which measure the non-productness in case of quantum state vectors. 
The \emph{negativity} \cite{VidWer02} is such a measure. 
Now, for a non-trivial bipartition $b$ and normalized tensor 
$X\in V$ (i.e., $\fNorm{X}=1$) 
\begin{equation}
  \frac 12 \left( (\OneNorm{X}^{b})^2-1\right)
  =
  \frac 12 \left( \OneNorm{  \left( \vect(X){\vect(X)}^{T} \right)^{T_b}  }-1 \right)
\end{equation}
is the negativity \cite{VidWer02} of the quantum state vector $X$ w.r.t.\ the bipartition $b$, 
where $T_{b}$ denotes the partial transposition w.r.t.\ $b$ and $\vect(X)$ the vectorization of $X$, i.e. the $[L]$-matrization. 

Theorem~\ref{thm:generic:lowerbound:max} applies to tensor recovery with the regularizer~\eqref{eq:max_norm_tensors}. 
We illustrate the lower bound for the special case of $4$-way tensors ($L=4$) with equal local dimensions $n_i = n$ and a regularizer norm given by $\mc B = (\{1,2\},\{1,3\}, \{1,4\})$. 
Then the critical number of measurements \eqref{eq:min_nr_of_measurements} in the lower bounds~\eqref{eq:generic:lowerbound:max} and \eqref{eq:deltaDC_bound} is 
\begin{equation}\label{eq:lower_bound_funny3}
  \kappa 
  = 
  \delta(\DC(\regNormMax[\mu^\ast]{\argdot};x_0)) 
  \approx 
  \min_{b \in \mc B}  \rank^{\mathrm{eff}}_{b}(X_0) \, n^2 . 
\end{equation}
If $X_0$ is a tensor product, this becomes $\kappa \approx n^2$. 

%%% --------------------------------------------
\subsection{RIP for the HOSVD and TT ranks}
%%% --------------------------------------------
A similar statement as Theorem~\ref{thm:RIP_low_rank_and_sparse} has been proved for the HOSVD and TT rank for the case of sub-Gaussian measurements by Rauhut, Schneider, and Stojanac \cite[Section~4]{RauSchSto16}. 
They also showed that the RIP statements lead to a partial recovery guarantee for iterative hard thresholding algorithms. 
Having only suboptimal bounds for TT and HOSVD approximations has so far prevented proofs of full recovery guarantees. 

It is unclear how RIP results could be extended to the ``ranks'' without an associated tensor decomposition and probably these ranks need to be better understood first.  

%%% --------------------------------------------
\subsection{Numerical experiments}
%%% --------------------------------------------
We sample the statistical dimension \eqref{eq:stad_dim_as_dist} numerically for $L=4$ instances of the max-norm \eqref{eq:max_norm_tensors} and a unit rank signal $X_0$;
see Figure~\ref{fig:meanSquaredDist-rank1}, 
where we have estimated the statistical dimension using the program \eqref{eq:generic:Gaussian_dist_max_SDP} with the dual norms being spectral norms of the corresponding $b$-matrizations. 
The numerical experiment suggest that the actual reconstruction behaviour of the $\mc B_3$-max-norm is close to \textbf{twice} the lower bound from Theorem~\ref{thm:generic:lowerbound:max}, which is given by $n^2$. 
The missing factor of two might be due to the following mismatch. 
In the argument with the circular cones we only have considered tensors of unit $b$-rank whereas the actual descent cone contains tensors of $b$-rank $2$ for some $b \in \mc B_3$. 
This discrepancy should lead to the lower bound be too low by a factor of $2$ which is compatible with the plots in Figure~\ref{fig:meanSquaredDist-rank1}. 

A similar experiment can be done for the similar sum-norm from \eqref{eq:def_max_sum_reg_norm}.  
This leads to similar statistical dimensions except for the tensor train bipartition, where the statistical dimension is significantly larger ($\sim 25\%$) for the sum-norm. 

\begin{figure}[t]
  \centering
  \iffrontiers\else
  \begin{adjustwidth}{-1.5cm}{-1.5cm}
  \fi
  \includegraphics[width = .49\linewidth]{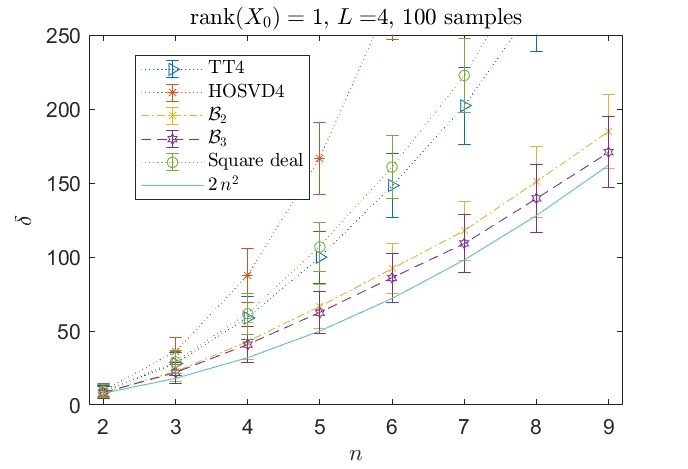}
  \includegraphics[width = .49\linewidth]{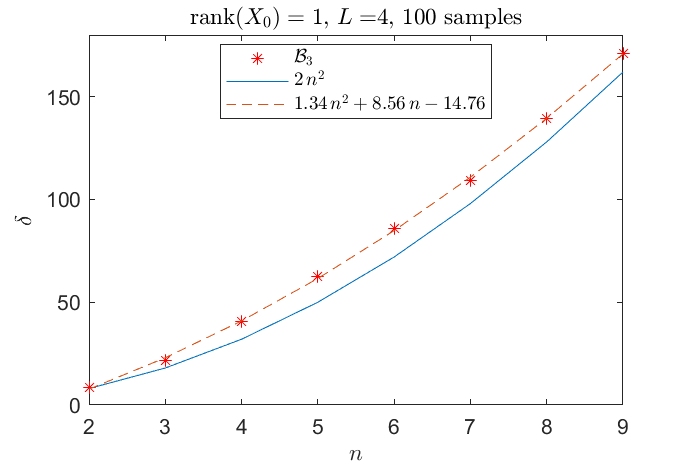}
  \iffrontiers\else
  \end{adjustwidth}
  \fi
  \caption{%
  \label{fig:meanSquaredDist-rank1}%
  Observed average of the statistical dimension \eqref{eq:stad_dim_as_dist} for a product signal $X_0 \in (\CC^n)^{\otimes k}$ with $L=4$ and the max-norm \eqref{eq:max_norm_tensors} as regularizer. 
  The norm is given by the bipartitions
  corresponding to 
  (i) the TT decomposition, 
  (ii) the HOSVD decomposition, 
  (iii) $\mc B_2 \coloneqq ( \{1,2\}, \{1,3\} )$, 
  $\mc B_3 \coloneqq ( \{1,2\}, \{1,3\}, \{1,4\} )$, 
  and $\mc B_{\mathrm{Square\; deal}} \coloneqq (\{1,2\})$.
  The statistical dimension $\delta$ corresponds to the critical number of measurements where the phase transition in the reconstruction success probability occurs. 
  The plots suggest that for the $\mc B_3$-max-norm the number of measurements scales roughly as twice the lower bound given by \eqref{eq:lower_bound_funny3}. 
  For the numerical implementation, the SDP~\eqref{eq:generic:Gaussian_dist_max_SDP} has been used. 
  The error bars indicate the unbiased sample standard deviation. 
  The numerics has been implemented with Matlab+CVX+SDPT3. 
  }

\end{figure}

\section{Conclusion and outlook}
We have investigated the problem of convex
recovery of simultaneously structured objects from few random
observations. 
We have revisited the idea of taking convex combinations of regularizers and have focused on the best among them, which is given by an optimally weighted maximum of the individual regularizers. 
We have extended and lower bounds on the required number of measurements by Mu et al.~\cite{MuHuaWri13} to this setting. 
The bounds are simpler and more explicit than those obtained by Oymak et al.~\cite{Oymak2015} for simultaneously low rank and sparse matrices. 
They show that it is not possible to improve the scaling of the optimal sampling rate in convex recovery even if optimal tuning and the maximum of simultaneous regularizers is
used, the latter giving the largest subdifferential. 
In more detail, we have obtained lower bounds for the number of measurements in the generic
situation and applied this to the cases of (i) simultaneously sparse and low-rank matrices and (ii) certain tensor structures. 

For these settings, we have compared the lower bounds to numerical experiments. 
In those experiments we have (i) demonstrated the actual recovery and (ii) estimated the statistical dimension that gives the actual value of the phase transition of the recovery rate. 
The latter can be achieved by sampling over certain SDPs. 
For tensors, we have observed that the lower bound can be quite tight up to a factor of $2$. 

The main question, whether or not one can derive strong rigorous recovery guarantees for efficient reconstruction algorithms in the case of simultaneous structures remains largely open. 
However, there are a few smaller questions that we would like to point out. 

Numerically, we have observed that if the weights deviate from the optimal ones has a relatively small effect for the sum of norms as compared to the maximum of norms. 
Indeed, $\delta(\mu) \coloneqq \delta(\cone( \partial \regNormSum{\argdot}(x_0) )^\circ )$ is a continuous function of $\mu$ whereas it appears to be non-continuous for $\regNormMax[\mu]{\argdot}$ due to \eqref{eq:cone_subdiff:max}. Of course it would be good to have tight upper bounds the both regularizers. 
Maybe, one can also find a useful interpolation between $\regNormSum[\mu]{\argdot}$ and $\regNormMax[\mu]{\argdot}$ by using an $\ell_p$ norm of the vector containing the single norms $\mu^\ast_i \norm{\argdot}_{(i)}$. 
This interpolation would give the max-norm $\regNormMax[\mu]{\argdot}$ for $p = \infty$ and the sum-norm $\regNormSum[\mu]{\argdot}$ for $p = 1$ and one could choose $p$ depending on how accurately one knows the optimal weights $\mu^\ast$. 
Finally, maybe one can modify an iterative non-convex procedure for solving the optimization problem we are using for the reconstructions such that one obtains recovery from fewer measurements. 

%%%=============================================
\section*{Acknowledgments}
%%%=============================================
We thank Micha{\l} Horodecki, Omer Sakarya, 
David Gross, 
Ingo Roth, 
Dominik Stoeger, and 
\v{Z}eljka Stojanak 
for fruitful discussions.

The work of MK was funded by the National Science Centre, Poland within the project Polonez (2015/19/P/ST2/03001) which has received funding from the European Union's Horizon 2020 research and innovation programme under the Marie Sk{\l}odowska-Curie grant agreement No 665778. The work of SJS was partially supported by the grant DMS-1600124 from the National Science Foundation (USA). 
PJ has been supported by DFG grant JU 2795/3.

%%%=============================================
\iffrontiers
\bibliographystyle{frontiersinHLTH_FPHY} 
\else
\bibliographystyle{./myapsrev4-1}
\fi
% \bibliography{../manuscript/merged_2018-08.bib, ../manuscript/new.bib}
\bibliography{new,combined}
%%%=============================================

%%%========================================================
%%% ===================== bbl file content ================
%%%========================================================
% none
\end{document}